%% file: main.tex
\documentclass[11pt]{article}
\input{defs}

\usepackage{fullpage}
\usepackage{orcidlink}
\usepackage{xcolor}
\usepackage{bbm}
\usepackage{subfig}
\usepackage{graphicx}

\newcommand{\mprobst}[1]{{\bf \color{Red} Max: #1}}

\DeclareMathOperator{\vol}{vol}

\newcommand\dir{\overrightarrow}
\DeclareMathOperator{\supp}{supp}

\title{Near-Optimal Algorithm for Directed Expander Decompositions}

\date{}
\newcommand*\samethanks[1][\value{footnote}]{\footnotemark[#1]}
\author{Aurelio L. Sulser\thanks{The research leading to these results has received funding from the starting grant “A New Paradigm for Flow and Cut Algorithms” (no. $TMSGI2\_218022$) of the Swiss National Science Foundation.} \\ ETH Zurich \\ asulser@ethz.ch \\ \and Maximilian Probst Gutenberg\samethanks \hspace{0.5em}\orcidlink{0000-0003-3522-156X}
 \\ ETH Zurich \\ maximilian.probst@inf.ethz.ch}

\begin{document}

\pagenumbering{gobble}

\maketitle

\begin{abstract}
%In this work, we introduce the first algorithm to compute expander decompositions in $m$-edge \emph{directed} graphs with near-optimal time complexity of $\tilde{O}(m)$\footnote{In this article, we use $\tilde{O}(\cdot)$ notation to suppress factors logarithmic in $m$, i.e., $O(m \log^c m) = \tilde{O}(m)$ for any constant $c > 0$.}. Additionally, our algorithm can maintain such decompositions in dynamic graphs, achieving near-optimal update times. This significantly improves upon previous algorithms \cite{bernstein2020deterministic, hua2023maintaining}, which were optimal only up to subpolynomial factors. Furthermore, our approach is notably simpler and more accessible than prior works.

%To develop our new algorithm, we introduce a novel push-pull-relabel flow framework that generalizes the classic push-relabel flow algorithm \cite{goldberg1988new} which was recently dynamized for computing expander decompositions in \emph{undirected} graphs \cite{henzinger2020local, saranurak2019expander}. We then show that the flow problems formulated in recent work \cite{hua2023maintaining} to decompose directed graphs can be solved much more efficiently in the push-pull-relabel flow framework. 

%Our algorithm has already been employed to achieve the currently fastest algorithm for computing min-cost flows \cite{vdB2024decrMincost}. We further believe that our approach can expedite and simplify recent breakthroughs in combinatorial graph algorithms, particularly those aimed at developing fast maximum flow algorithms \cite{chuzhoy2024faster, chuzhoy2024maximum, bernstein2024maximum}.

In this work, we present the first algorithm to compute expander decompositions in an $m$-edge \emph{directed} graph with near-optimal time $\tilde{O}(m)$\footnote{In this article, we use $\tilde{O}(\cdot)$ notation to suppress factors logarithmic in $m$, i.e. $O(m \log^c m) = \tilde{O}(m)$ for every constant $c > 0$.}. Further, our algorithm can maintain such a decomposition in a dynamic graph and again obtains near-optimal update times. Our result improves over previous algorithms \cite{bernstein2020deterministic, hua2023maintaining} that only obtained algorithms optimal up to subpolynomial factors.

In order to obtain our new algorithm, we present a new push-pull-relabel flow framework that generalizes the classic push-relabel flow algorithm \cite{goldberg1988new} which was later dynamized for computing expander decompositions in \emph{undirected} graphs \cite{henzinger2020local, saranurak2019expander}. We then show that the flow problems formulated in recent work \cite{hua2023maintaining} to decompose directed graphs can be solved much more efficiently in the push-pull-relabel flow framework. 

Recently, our algorithm has already been employed to obtain the currently fastest algorithm to compute min-cost flows \cite{vdB2024decrMincost}. We further believe that our algorithm can be used to speed-up and simplify recent breakthroughs in combinatorial graph algorithms towards fast maximum flow algorithms \cite{chuzhoy2024faster, chuzhoy2024maximum, bernstein2024maximum}.
\end{abstract}

\pagebreak

%\tableofcontents

\pagebreak
\pagenumbering{arabic}

\section{Introduction}
Over the past two decades, expanders and expander decompositions have been pivotal in advancing on fundamental algorithmic graph problems. The development and application of the first fast algorithm to compute near-expander decompositions was given in the development of the first near-linear time Laplacian solvers \cite{spielman2004nearly}, a breakthrough in modern graph algorithms. Subsequently, a line of research \cite{henzinger2020local, wulff2017fully, nanongkai2017dynamic,nanongkai2017dynamicMinimum} has focused on strengthening this result by developing fast flow-based pruning techniques that refine near-expander decompositions into expander decompositions. This line of research culminated in \cite{saranurak2019expander} where a new, faster, simpler, and more user-friendly expander decomposition framework was presented. This advancement has catalyzed the widespread use of expander decompositions as a tool in graph algorithms and was instrumental in the recent surge of applications of expander decompositions in both static and dynamic graph settings for various cut, flow, and shortest path problems \cite{spielman2004nearly, kelner2014almost, henzinger2020local, wulff2017fully, nanongkai2017dynamic,nanongkai2017dynamicMinimum,chuzhoy2019new,bernstein2020deterministic,liu2020vertex, bernstein2020fully,van2021minimum,saranurak2021simple,chalermsook2021vertex, li2021deterministic,chuzhoy2021deterministic, goranci2021expander, chuzhoy2021decremental,bernstein2022deterministic, bernstein2022deterministic, kyng2022derandomizing, jin2022fully, van2023deterministic, kyng2023dynamic, chen2023almost, jin2024fully, chuzhoy2024maximum, bernstein2024maximum, vdB2024decrMincost}. In this work, we study the problem of computing and maintaining expander decompositions in \emph{directed} graphs, defined as follows. 

\begin{definition}[Directed Expander Decomposition]\label{def:expdecomIntro}
Given an $m$-edge directed graph $G$, we say a partition $\mathcal{X}$ of the vertex set of $G$ and a subset of edges $E^r \subseteq E$ forms an $(\beta, \phi)$-expander decomposition if 
\begin{enumerate}
    \item \label{Def:ED-item1} $\forall X \in \mathcal{X}$, $G[X]$ is a $\phi$-expander meaning for all cuts $(S, \bar{S}): \frac{\min\{e_G(S, \bar{S}), e_G(\bar{S}, S)\}}{\min\{\vol_G(S), \vol_G(\bar{S})\}} \geq \phi$, and
    \item \label{Def:ED-item2} $|E^r| \leq \beta \cdot \phi \cdot m$, and 
    \item \label{Def:ED-item3} the graph $(G \setminus E^r) / \mathcal{X}$, that is the graph $G$ minus the edges in $E^r$ where expander components in $\mathcal{X}$ are contracted into supernodes, is a directed acyclic graph (DAG).
\end{enumerate}
\end{definition}

In our algorithm, we implicitly maintain an ordering of the partition sets in $\mathcal{X}$ and let $E^r$ be the edges that go 'backward' in this ordering of expander components. Note that we can only obtain a meaningful bound on the number of such 'backward' edges since a bound on \emph{all} edges between expander components cannot be achieved as can be seen from any graph $G$ that is acyclic (which implies that $\mathcal{X}$ has to be a collection of singletons by  \Cref{Def:ED-item1} forcing all edges to be between components). Directed expanders and expander decompositions have been introduced in \cite{bernstein2020deterministic} in an attempt to derandomize algorithms to maintain strongly connected components and single-source shortest paths in directed graphs undergoing edge deletions. Recently, \cite{hua2023maintaining} gave an alternative algorithm to compute and maintain directed expander decomposition that refines the framework from \cite{bernstein2020deterministic} and heavily improves subpolynomial factors. Besides working on directed graphs, this algorithm also yields additional properties for expander decompositions that cannot be achieved with existing techniques - even in undirected graphs. In this article, we further refine these techniques to obtain an algorithm that is optimal up to logarithmic factors in $m$ - as opposed to subpolynomial factors. Since their invention, directed expander decompositions and the techniques to maintain such decompositions have been pivotal in the design of fast dynamic graph data structures and in ongoing research for fast 'combinatorial' maximum flow algorithms. We discuss in \Cref{subsec:appl} these recent lines of research and how our algorithm benefits recent breakthrough results. For an in-depth discussion of expander decomposition techniques and applications both in directed and undirected graphs, we refer the interested reader to \Cref{subsec:prevWork}.

\subsection{Our Contribution} 
In this article, we finally give a simple algorithm that generalizes the algorithm from \cite{saranurak2019expander} in a clean way. Further, our algorithm is the first to obtain near-optimal runtimes for both static and dynamic expander decompositions in directed graphs. Our result is summarized in the theorem below.

\begin{theorem}\label{Main-thm}Given a parameter $\phi \leq c/ \log^{12} m$ for a fixed constant $c > 0$, and a directed $m$-edge graph $G$ undergoing a sequence of edge deletions, there is a randomized data structure that constructs and maintains a $(O(\log^{19} m), \Omega(\phi/\log^{12} m), O(1/\log^8 m))$-expander decomposition \footnote{Here we use the augmented expander decomposition definition \ref{def:augmentedED}.} $(\mathcal{X}, E^r)$ of $G$. The initialization of the data structure takes time $O(m \log^{20}(m)/\phi)$ and the amortized time to process each edge deletion is $O(\log^{28}(m)/\phi^2)$.
\end{theorem} 

Further, our algorithm has the property that it is refining for up to $O(\phi \cdot \psi \cdot m)$ edge deletions meaning that $\mathcal{X}$ is a refinement of its earlier versions (every expander component in $\mathcal{X}$ is a subset of an expander component in any earlier expander decomposition) and the size of $E^r$ does never exceed $\tilde{O}(\phi m)$. Our algorithms are deterministic, however, they rely on calling a fast randomized algorithm to find balanced sparse cuts or certify that no such cut exists (see \cite{khandekar2009graph, louis2010cut}). Our new techniques are much simpler and more accessible than previous work, besides also being much faster. We hope that by giving a simpler algorithm for directed expander decompositions, we can help to make this tool more accessible to other researchers in the field with the hope that this can further accelerate recent advances in dynamic and static graph algorithms.

\subsection{Applications}
\label{subsec:appl}

Our new algorithms have direct applications to the currently fastest approaches for bipartite matching/ maximum flow/ min-cost flows and the data structures that are employed to obtain these results:
\begin{enumerate}
    \item Our algorithm is already used in the fastest min-cost flow algorithm that is known to-date \cite{vdB2024decrMincost} which achieves runtime $m \cdot e^{O(\log^{3/4}(m) \log\log(m))}$ yielding the first improvement over the recent breakthrough in \cite{chen2022maximum} achieving the first near-linear time algorithm for min-cost flows. In the framework of \cite{vdB2024decrMincost}, our algorithmic techniques are used to maintain an expander decomposition of an \emph{undirected} graph where it is heavily exploited that our algorithm maintains the expander decomposition such that it refines over time. This guarantee is pivotal in the construction of a fully-dynamic algorithm to maintain dynamic expander hierarchies which is the key data structure in the paper. While \cite{vdB2024decrMincost} could also have relied solely on the techniques \cite{hua2023maintaining} to obtain such a data structure with subpolynomial update and query times, these subpolynomial factors would have been substantially larger and would thus not have yielded a faster min-cost flow algorithm overall.

    \item In \cite{bernstein2020deterministic},  directed expander decompositions for decremental graphs were used to obtain the first algorithm to maintain $(1+\epsilon)$-approximate Single-Source Shortest-Paths (SSSPs) in a decremental graphs in time $o(mn)$, though only for the special case of dense graphs. This problem is well-motivated as a simple reduction based on the Multiplicative Weights Update (MWU) framework implies that the maximum flow problem can be solved approximately by solving approximate decremental SSSP. By standard refinement of flows this yields an exact maximum flow algorithm. Very recently, Chuzhoy and Khanna \cite{chuzhoy2024faster, chuzhoy2024maximum} showed that for the update sequence generated by the MWU to solve the bipartite matching problem - a special case of the maximum flow problem - decremental SSSP can be maintained in $n^{2+o(1)}$ time by refining the techniques from \cite{bernstein2020near, bernstein2020deterministic}. This yields the first near-optimal 'combinatorial'\footnote{We refrain from defining the scope of combinatorial algorithms here and refer the reader to \cite{chuzhoy2024faster, chuzhoy2024maximum} for a discussion.} algorithm for the bipartite matching problem for very dense graphs. While the above algorithms are 'combinatorial', they are still very intricate. Most of these complications stem from the maintenance of the directed expander decomposition used internally by the decremental SSSP data structure. We hope that our technique can help to simplify and speed-up these components to yield a simpler algorithm overall.

    \item In independent work \cite{bernstein2024maximum}, an alternative 'combinatorial' maximum flow that runs in $n^{2+o(1)}$ time was given. This algorithm cleverly extends push-relabel algorithms to run more efficiently when given an ordering of vertices that roughly aligns with the topological order induced by the acyclic graph formed from the support of an optimal maximum flow solution. To obtain this approximate ordering they compute a static expander hierarchy of the directed input graph. This generalizes the notion of directed expander decompositions further. In their work, they heavily build on the techniques from \cite{hua2023maintaining} to obtain the expander hierarchy. Unfortunately, the algorithm to obtain this hierarchy is very involved. We hope that our new techniques can help to simplify and speed-up their algorithms.    
\end{enumerate}

\subsection{Our Techniques} 

\paragraph{High-Level Strategy.} We obtain our result by following the high-level strategy of \cite{saranurak2019expander} for undirected graphs: we draw on existing literature (specifically \cite{khandekar2009graph, louis2010cut}) for an algorithm that either outputs a balanced sparse cut which allows us to recurse on both sides; or outputs a witness that no such cut exists. This witness can be represented as an expander graph $W$ that embeds into $G \cup F$ with low congestion where $F$ is a set of few \emph{fake} edges. In the second case, we set up a flow problem to extract a large expander (the first algorithm only finds balanced sparse cuts, so many unbalanced sparse cuts might remain) which suffices to again recurse efficiently.

\paragraph{The (Dynamic) Flow Problem in \cite{saranurak2019expander}.} To outline our algorithm, we first sketch the techniques of \cite{saranurak2019expander}. In \cite{saranurak2019expander}, the following sequence of flow problems is formulated: initially, we add $\frac{1}{2\phi}$ units of source commodity to each endpoint of an edge in $F$ and then ask to route the commodity in $G$ where each vertex $v$ is a sink of value $\deg_G(v)$ and each edge has capacity $\frac{1}{2\phi}$. It then runs an (approximate) maximum flow algorithm on the flow problem. Whenever the algorithm detects that no feasible flow exists\footnote{Technically, the algorithm might already output cuts when some cut has capacity less than a constant times the amount of flow that is required to be routed through the cut.}, it finds a cut $(A, \overline{A})$ where $A$ is the smaller side of the cut and then poses the same problem for the network $G[\overline{A}]$ where this time the source commodity is assigned for each edge in $E_G(A, \overline{A}) \cup F$. The algorithm terminates once the flow problem can be solved and outputs the final induced graph. In \cite{saranurak2019expander}, it is shown that once a feasible flow exists then the (induced) graph is a $\Omega(\phi)$-expander. Further, it is shown that in the sequence of flow problems, each problem can be warm-started by re-using the flow computed in the previous instance to detect a cut induced on the remaining vertex set. This result is obtained by two main insights: 
\begin{enumerate}
    \item if the flow $\ff$ to find the cut $(A, \bar{A})$ is a \emph{pre-flow}, that is a flow that respects capacities and has no negative excess at any vertex (i.e. it does not route away more flow from a vertex than is inputted by the source), then injecting additional source flow for any edge $E_G(A, \bar{A})$ guarantees that the induced flow $\ff|_{\bar{A}}$ is a pre-flow in the flow problem formulated for $G[\bar{A}]$. That is because the amount of flow that was routed via such a cut edge is at most $2/\phi$ and thus placing $2/\phi$ new source commodity at the endpoint ensures that no negative excess exists in the induced flow $\ff|_{\bar{A}}$,
    \item the classic push-relabel framework can naturally be extended to warm-start on such a flow $\ff|_{\bar{A}}$ as it is built to just further refine pre-flows at every step.
\end{enumerate}
This dynamization of the push-relabel framework allows to bound the cost of computation of \emph{all} flow problems linearly in the amount of source commodity which in term is bounded by the number of edges that appear in either $F$ or one of the identified min-cuts. Finally, \cite{saranurak2019expander} shows that the amount of source commodity remaining in $\bar{A}$ decreases over the sequence of flow problems proportional to the volume of the set $A$ of vertices that are removed at each step. Indeed, they observe that at each vertex $v$ in $A$ at least $\deg(v)$ many commodity units are absorbed and that the total amount of source injected due to the cut edges $E_G(A,\bar{A})$ is bounded by $\frac{1}{2 \phi} \cdot e_G(A,\bar{A}) \leq \vol_G(A)/2$. This yields that the final induced graph is still large. Thus, the final graph outputted is a large expander, as desired. 

\begin{figure}
  \centering
  \subfloat[In directed graphs, cuts are asymmetric. While $(A, \bar{A})$ might be a sparse cut, the cut $(\bar{A}, A)$ might contain many edges. A straightforward extension of \cite{saranurak2019expander} would inject $2/\phi$ units of commodity to each endpoint of $E_G(A, \bar{A}) \cup F$. However, it is not clear with this approach how to bound the total amount of flow injected throughout the algorithm.]{\includegraphics[width=0.45\textwidth]{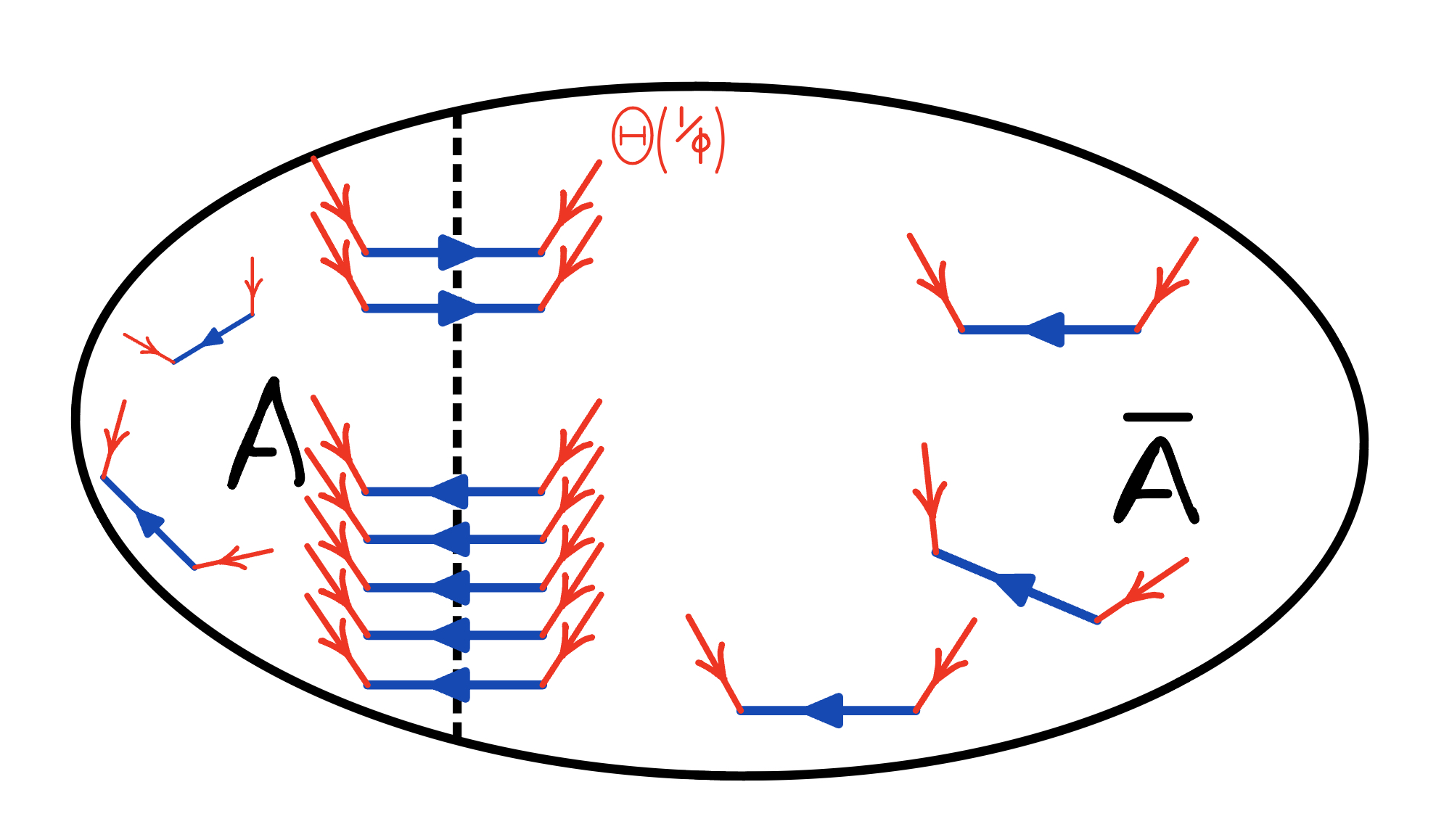}\label{fig:SourceInjection-A}}
  \hfill
  \subfloat[We inject $\Theta(1/\psi)$ units of commodity at the end points of any witness embedding path (green) going through an edge of $E_G(A, \bar{A}) \cup F$. We can bound the total amount by $O(\vol(A)/\text{poly}(\psi))$. But the injection might well be in the interior of $A, \bar{A}$ possibly leaving negative excess at the endpoints of the cut edges.]{\includegraphics[width=0.45\textwidth]{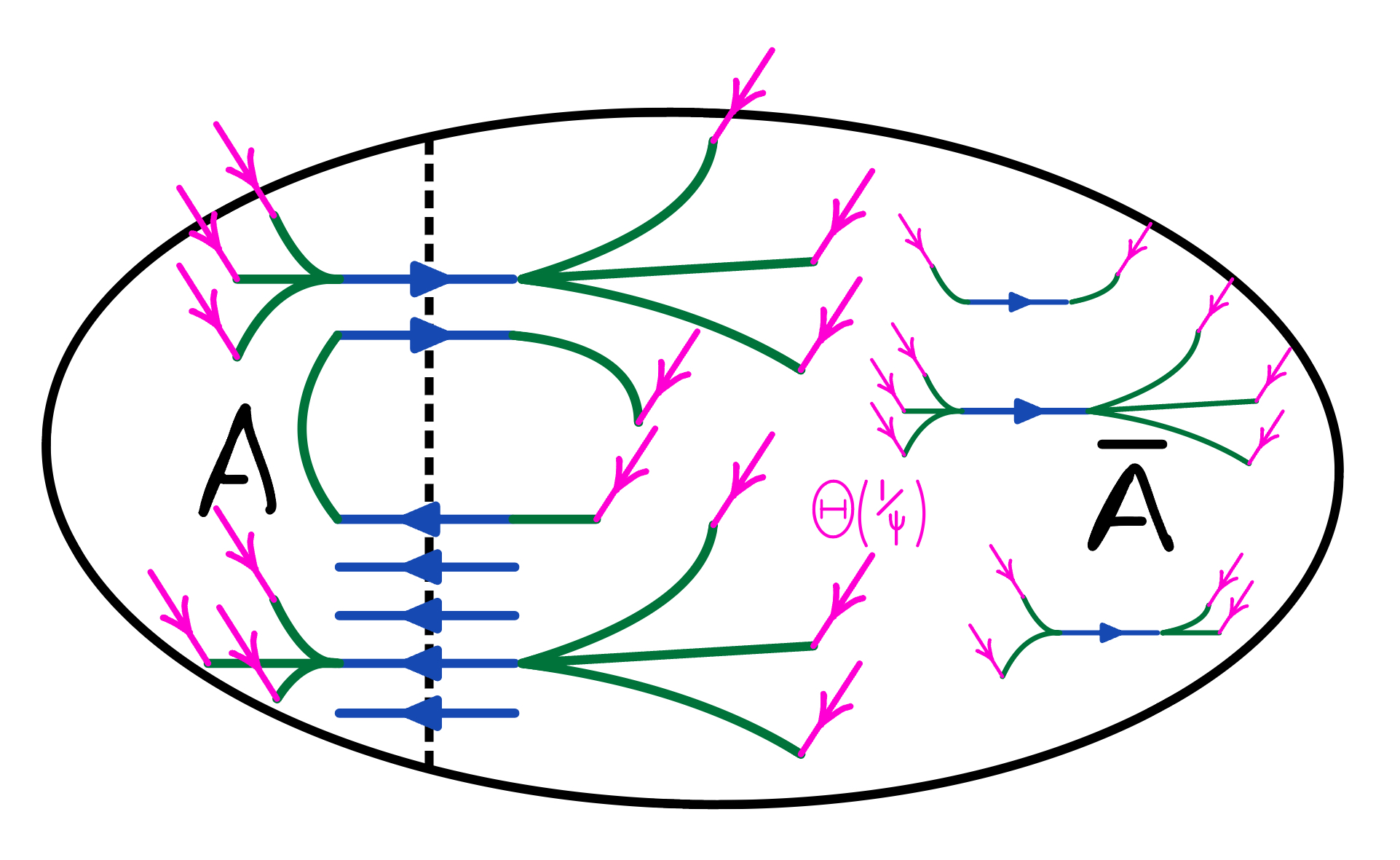}\label{fig:SourceInjection-B}}
  \caption{Injection of Commodity due to edges in $E_G(A, \bar{A}) \cup F$}
\end{figure}

\paragraph{The (Dynamic) Flow Problem in Directed Graphs.} In directed graphs, while the above flow problem upon becoming feasible also certifies that the remaining graph is a $\Omega(\phi)$-expander, the argument that the sequence of flow problems terminates does not work: the asymmetry of cuts might force us for a small cut $E_G(A, \bar{A})$ to induce on $\bar{A}$ while having many edges in $E_G(\bar{A}, A)$ each of which would add $2/\phi$ source flow to the new flow problem (see \Cref{fig:SourceInjection-A}). Hence, we might end up injecting up to $\Omega(\vol_G(A)/\phi)$ more flow due to the cut edges. This makes it seemingly impossible to argue that the amount of source commodity in the next flow problem is smaller. To recover the argument that the sequence of flow problems terminates (with the remaining expander graph being large) both \cite{bernstein2020deterministic, hua2023maintaining} suggest setting up the flow problems more carefully such that each cut $(A, \bar{A})$ that is found in this sequence and induced upon is a \emph{sparse cut}. Here, we only describe the less lossy flow problem formulation developed in \cite{hua2023maintaining}. To ensure that each cut $(A, \bar{A})$ that is found is a sparse cut, \cite{hua2023maintaining} proposes a slightly different flow problem: instead of adding source commodity $\frac{1}{2\phi}$ per endpoint of an edge that is fake or not fully contained in the induced graph, it tailors the amount of new source commodity using the witness graph $W$, possibly injecting much less source commodity in the process. Concretely, we have that $W$ is a $\psi$-expander over the same vertex set as $G$ with degrees similar to degrees in $G$ up to a factor of $\Theta(\text{poly}(\psi))$ and an embedding $\Pi$ into $G \cup F$ with congestion $O(\phi/\psi)$, for $\phi = \tilde{\Theta}(1)$. To set-up the flow problem, we inject $\Theta(1/\psi)$ units at the endpoints of any edge $e$ in the witness $W$ whose embedding path $\Pi(e)$ goes through an edge in $E_G(A, \bar{A}) \cup F$. We note that any such witness embedding path $\Pi(e)$ either crosses the sparse cut $(A, \bar{A})$ or has an endpoint in $A$. This allows to bound the additional source injected by $O(\vol_G(A)/\text{poly}(\psi)) = \tilde{O}(\vol_G(A))$ where we use that $\psi = \tilde{\Theta}(1)$. But while correctness and termination of the flow problem sequence are now ensured, this leaves a significant problem: the current flow $\ff$ that was used to find the cut $(A, \bar{A})$ no longer has the property that $\ff|_{\bar{A}}$ is a \emph{pre-flow} in the flow problem formulated on network $G[\bar{A}]$ even if $\ff$ is a pre-flow. While capacity constraints are still enforced, i.e. $\ff|_{\bar{A}}$ still is a pseudo-flow (see \cite{hochbaum2008pseudoflow} for reference on pseudo-flow), some vertices might now have negative excess since the additional commodity might be injected in the interior of $\bar{A}$ far from the cut $(A, \bar{A})$ (see \Cref{fig:SourceInjection-B}). Indeed, for an edge $(u,v)$ in $E_G(A, \bar{A})$ a lot of flow might have been routed through $(u,v)$ but no additional commodity might be injected at $v$ since all witness embedding paths passing through $(u,v)$ might not end in $v$. Thus, dynamizing the push-relabel framework does not appear natural for this sequence of problems as it crucially requires that the maintained flow is a pre-flow at all times. In \cite{hua2023maintaining}, an involved batching technique is used instead (based on the technique in \cite{nanongkai2017dynamicMinimum}) that does not use dynamic flow problems but instead reduces to few static flow problems, however, at the loss of quality and runtime by subpolynomial factors.

\paragraph{The Push-Pull-Relabel Framework.} The main technical contribution of this paper is a new framework that refines pseudo-flows as efficiently as the push-relabel framework refines pre-flows. Thus, we give a generalization of the latter widely-used and well-studied framework that we believe might have applications well beyond our problem. Recall that the classic push-relabel framework maintains labels $\bell$ for all vertices and a pre-flow $\ff$. In each iteration, it 'pushes' positive excess flow at a vertex $v$ to a vertex at a lower label (to be precise to a vertex at level $\bell(v) - 1$); or if no 'push' is possible, it increases the labels of some vertices: it 'relabels'. Using a clever potential-based analysis, one can show that it suffices to only increase the labels to a certain threshold before all flow is settled. In our framework, we allow deleting edges, without compensating by adding source commodity at the endpoints, which might create negative excess leaving $\ff$ a pseudo-flow (instead of a pre-flow). Now, while our framework applies the same strategy for 'pushes' and 'relabels', we also need a new operation 'pull'. Intuitively, our algorithm tries to 'pull' back the source commodity that now causes the negative excess (this unit of commodity was 'pushed' earlier to some other vertices). To do so, a vertex $v$ with negative excess can 'pull' commodity from vertices at a higher level (again it can only pull from a vertex at level $\bell(v) + 1$). But it is not difficult to construct an example where this strategy does not suffice: therefore, we also need to sometimes decrease the label of a vertex to ensure correctness. However, the latter change to the 'relabel' operation breaks the property that labels are non-decreasing over time. A property that is crucial in the existing efficiency analysis. Instead, we give a much more careful argument to analyze the potentials that bound the number of push, pull, and relabeling operations that deal with the non-monotonicity of the levels over time. The argument is sketched below. Combining this framework with the above-discussed set-up of dynamic flow problems as proposed in \cite{hua2023maintaining} then yields the first near-optimal algorithm to compute an expander decomposition in a directed graph. Further, our technique extends seamlessly to also deal with edge deletions to $G$, yielding an algorithm to prune expander graphs that undergo edge deletions.

\paragraph{A Sketch of the Runtime of Push-Pull-Relabel.} The run-time analysis in standard push-relabel considers the run-time contribution of the push and relabel operations separately. Using a potential argument, one can then relate the contribution of the push operations to the relabeling operations. A similar argument albeit much more delicate also allows to relate the contribution of the push and the pull operations to the contribution of the relabeling operations in the extended push-pull-relabel algorithm. What might seem more daunting is to bound the run-time contribution of the relabelings given that the level function $\bell$ is no longer point-wise non-decreasing. Let us revisit the argument to bound the run-time contribution of the relabelings in the push-relabel of \cite{saranurak2019expander}. Any relabeling of $v$, incurs a cost of $O(\deg(v))$ as each incident edge has to be checked before a relabeling. At such a relabeling of $v$, the sink of $v$, which has by choice capacity $\deg_G(v)$, must be full and any commodity unit in it remains there till termination. Since the label of $v$ does not decrease and is bounded by $h$, we may thus charge for the run-time contribution of the relabelings of $v$ each unit in the sink of $v$ exactly $h$. Overall vertices, we conclude that any commodity unit was charged at most $h$ units. In the push-pull-relabel algorithm, a commodity unit might end up in various sinks. So a more clever charging argument needs to be devised. To guide intuition an analogy to flows of protons and electrons is useful. The protons correspond to commodity units and the electrons to the lack of commodity units. We are now moving protons and electrons around in the network and whenever a proton and an electron meet at a vertex they form a neutron which stays put at the vertex indefinitely. The neutrons will provide a mean to charge work in the analysis. The sink of each vertex $v$ can absorb $\deg(v)$ commodity units or put differently $\deg(v)$ protons. Electrons only form when we delete an edge $(u,v)$, i.e. exactly $\ff^+(u,v)$ new electrons form at $v$. By choice, we now say that if we perform a push from $u$ to $v$ of one unit then exactly one proton moves from $u$ to $v$, while if we perform a pull from $u$ to $v$ then exactly one electron moves from $v$ to $u$. Since we only perform a push away from $v$ if there are more than $\deg_G(v)$ free protons (not part of a neutron) at $v$ and a pull towards $v$ if there are free electrons at $v$, we find that $\min(\textbf{n}(v) + \deg_G(v)/2, \textbf{p}(v))$, where $\textbf{p}(v)$ denotes the number of protons at $v$ (including the ones in a neutron) and $\textbf{n}(v)$ denotes the number of electrons at $v$, is non-decreasing. This fact allows us to conclude that whenever we change from increasing $\bell(v)$ to decreasing $\bell(v)$ at least $\deg(v)/2$ new neutrons have formed at $v$. Since any neutron at $v$ stays put indefinitely, we can charge the run-time contributions of the relabelings of $v$ to the number of neutrons at $v$. 

\paragraph{Roadmap.} In the remainder of the article, we first give preliminaries in \Cref{sec:prelim}, then present our new push-pull-relabel framework in \Cref{sec:pushpullRelabel} and finally show how to obtain our result in \Cref{Main-thm} using the new framework in \Cref{sec:dirExpanderDecomp}.

\section{Preliminaries}
\label{sec:prelim}

\paragraph{Graphs.} We let $\deg_G \in \R^{V(G)}$ denote the degree vector of graph $G$. For all vertex $v \in V$, we have $\deg_G(v)$ equal to the number of edges incident to $v$ (both incoming and outgoing are counted). Moreover, for any subset of edges $D \subseteq E(G)$ we denote by $G_{D}$ the graph with vertex set $V(G)$ and edge set $D$ and by $\deg_D$ the degree vector of the subgraph $G_{D}$. We denote by $\vol_G(S)$ for any $S \subseteq V$, the sum of degrees of vertices in $S$. We denote by $E_G(A, B)$ for any $A, B \subseteq V$ the set of directed edges in $E(G)$ with tail in $A$ and head in $B$. We define $e(G) = |E(G)|$ and $e_G(A,B) = |E_G(A,B)|$. For any partition $\mathcal{P}$ of $V(G)$, we denote the graph obtained by contracting each partition class to a single vertex by $G/\mathcal{P}$. Two vertices in $G/\mathcal{P}$ are adjacent if there is an edge between the corresponding partition classes in $G$. Moreover, for any vertex $v \in V(G)$ we denote by $\mathbbm{1}_v \in \R^{V(G)}$ the vector with all entries equal to zero apart from the entry at $v$ equaling one. %\aurelio{increasing partitions and edge set}

\paragraph{Flows.} We call a tuple $(G, \cc, \Delta, \nabla)$ a flow problem, if $G$ is a directed graph, the capacity function $\cc: V(G) \times V(G) \rightarrow \R^{\geq 0}$ is such that for all $(u,v) \not \in E$ we have $\cc(u,v) = 0$, and $\Delta, \nabla: V(G) \rightarrow \R^{\geq 0}$ denote the source and the sink capacities.
We denote flows on $G$ as functions $\ff: V(G) \times V(G) \rightarrow \R$ such that $\ff$ is anti-symmetric, i.e. $\ff(u,v) = - \ff(v,u)$. Given a vertex $v \in V(G)$ we introduce the notation $\ff(v) = \sum_u \ff(v,u)$ and likewise $\cc(v) = \sum_u \cc(v,u) + \cc(u,v).$ Moreover, we write $\ff^+(u,v) = \max(\ff(u,v), 0)$. Given a flow $\ff$, we say a vertex $v \in V(G)$ has $\Gamma(v) = \Delta(v) - \ff(v)$ excess. We say that it has positive excess if $\Gamma(v) > \nabla(v)$ and $v$ has negative excess if $\Gamma(v) < \nabla(v)/2.$ For a subset $\tilde{V} \subseteq V(G)$, we induce the flow $\ff$, and the sink $\nabla$, source $\Delta$ and edge capacities $\cc$ onto $G[\tilde{V}]$ in a function sense and write $\ff|_{\tilde{V}}, \nabla|_{\tilde{V}}, \Delta|_{\tilde{V}}, \cc|_{\tilde{V}}$. Moreover, for any subset $E \subseteq E(G)$, we write $\ff|_{E}$ for the flow induced by $\ff$ onto the subgraph of $G$ consisting only of the edges $E$. We say that a flow $\ff$ is a \emph{pseudo-flow} if it satisfies the capacity constraints:
\[\forall (u,v) \in V(G) \times V(G): -\cc(v,u) \leq \ff(u,v) \leq \cc(u,v).\] We say $\ff$ is a \emph{pre-flow} if $\ff$ is a pseudo-flow and has no negative excess at any vertex. We say a flow $\ff$ is \emph{feasible} if it is a pre-flow and additionally no vertex has positive excess. Moreover, for any subset $S \subseteq V$ we denote by $\Delta(S), \nabla(S)$ the sum $\sum_{v \in S} \Delta(v), \sum_{v \in S} \nabla(v)$ respectively and for any subset $D \subseteq E(G)$ we denote by $\cc(D)$ the sum $\sum_{d \in D} \cc(d).$

%\paragraph{Flows.} We call a tuple $(G, \cc, \Delta, \nabla)$ a flow problem, if $G$ is a directed graph, the capacity function $\cc: V(G) \times V(G) \rightarrow \R^+$ is such that for all $(u,v) \not \in E$ we have $\cc(u,v) = 0$, and $\Delta, \nabla: V(G) \rightarrow \R^+$ denote the source and the sink capacities. We denote flows on $G$ as functions $\ff: V(G) \times V(G) \rightarrow \R$ such that for all $(u,v) \not \in E$ we have $\ff(u,v) = 0$. Given a vertex $v \in V(G)$ we introduce the notation $\ff(v) = \sum_u \ff(v,u) + \ff(u,v), \ff^+(v) = \sum_u \ff(v,u), \ff^-(v) = \sum_u \ff(u,v)$ and likewise $\cc(v) = \sum_u \cc(v,u) + \cc(u,v), \cc^+(v) = \sum_u \cc(v,u), \cc^-(v) = \sum_u \cc(u,v).$ For a subset $\tilde{V} \subseteq V(G)$, we induce the flow $\ff$, and the sink $\nabla$, source $\Delta$ and edge capacities $\cc$ onto $G[\tilde{V}]$ in a function sense and write $\ff|_{\tilde{V}}, \nabla|_{\tilde{V}}, \Delta|_{\tilde{V}}, \cc|_{\tilde{V}}$.

\paragraph{Expanders.} 
Given graph $G = (V,E)$, we say a cut $(S, \bar{S})$ is $\phi$-out sparse if $e_G(S) \leq e(G)$ and $e_G(S, V \setminus S) < \phi \cdot \vol_G(S)$. We say $G$ is a $\phi$-out expander if it has no $\phi$-out-sparse cut. We say $G$ is a $\phi$-expander if $G$ and $G^{rev}$, the graph where all edges of $G$ are reversed, are both $\phi$-out expander. The next lemma that is folklore and crucial in our expander pruning argument.

\begin{lemma} \label{lm:helper}
    Given a $\phi$-expander $G=(V,E)$, then take $S \subseteq V$ and a set of edge deletions $D$. We have that $e_{G \setminus {D}}(S, V \setminus S)  < \frac{\phi}{4} \cdot \operatorname{vol}_G(S)$ implies $\min(\operatorname{vol}_G(S), \operatorname{vol}_G(V \setminus S)) < \frac{4\cdot |D|}{3\phi}.$
\end{lemma}
\begin{proof}
    If $\min(\operatorname{vol}_G(S), \operatorname{vol}_G(V \setminus S)) \geq \frac{4 \cdot |D|}{3 \phi}$, then $e_{G \setminus {D}}(S, V \setminus S) \geq \phi \cdot \min(\operatorname{vol}_G(S), \operatorname{vol}_G(V \setminus S)) - |{D}| \geq \frac{\phi}{4} \min(\operatorname{vol}_G(S), \operatorname{vol}_G(V \setminus S))$.
\end{proof}

\paragraph{Graph Embeddings.} Given two graphs $H$ and $G$ over the same vertex set $V$, we say that $\Pi$ is an embedding of $H$ into $G$ if for every edge $e = (u,v) \in E(H)$, $\Pi(e)$ is a simple $uv$-path in $G$. We define the congestion of an edge $e \in E(G)$ induced by the embedding $\Pi$ to be the maximum number of paths in the image of $\Pi$ that contain $e$. We define the congestion of $\Pi$ to be the maximum congestion achieved by any such edge $e \in E(G)$. 
Moreover, for any edge $e \in E(G)$ we will denote by $\Pi^{-1}(e)$ the set of edges $f \in E(H)$ such that $e$ is an edge on the path $\Pi(f)$. Given an entire set of edges $D \subseteq E(G)$, we denote by $\Pi^{-1}(D)$ the set of edges $f \in E(H)$ such that some edge of the path $\Pi(f)$ is in $D$. Given two graphs $H_1, H_2$ on the same vertex set and embeddings $\Pi_1 : H_1 \rightarrow G, \Pi_2 : H_2 \rightarrow G$ then we denote by $\Pi_1 \cup \Pi_2 : H_1 \cup H_2 \rightarrow G$ the embedding of the graph $H_1 \cup H_2 = (V(H_1), E(H_1) \cup E(H_2))$ .

\paragraph{Expander Decompositions with Witnesses.} For the rest of the article, we use a definition of expander decompositions that encodes much more structure than given in \Cref{def:expdecomIntro}. In particular, the definition below incorporates the use of witness graphs which are instrumental to our algorithm. 

\begin{definition}\label{def:Witness}
    Given a directed graph $G$, we say $(W, \Pi)$ is a $(\phi, \psi)$-out-witness for $G$ if 1) $W$ is a $\psi$-out-expander, and 2) $\Pi$ embeds $W$ into $G$ with congestion at most $\frac{\psi}{\phi}$, and 3) $ \forall v \in V(W): \deg_{G}(v) \leq \deg_{W}(v) \leq \frac{\deg_G(v)}{\psi}$. If $(W, \Pi)$ is a $(\phi, \psi)$-out-witness for $G$ and for $G^{rev}$, then we say that $(W, \Pi)$ is a $(\phi, \psi)$-witness for $G$.

\begin{fact}
    If $(W_1, \Pi_1)$ is a $(\phi, \psi)$-out-witness for $G$ and $(W_2, \Pi_2)$ is a $(\phi, \psi)$-out-witness for $G^{rev},$ then $(W_1 \cup W_2, \Pi_1 \cup \Pi_2)$ is a $(\psi\phi/4, \psi^2/2)$-witness for $G$.
\end{fact}

\begin{proof}
    We observe that for any cut $(S, V\setminus S): e_{W_1 \cup W_2}(S,V \setminus S) \geq \psi \cdot \min(\vol_{W_1}(S), \vol_{W_2}(S), \vol_{W_1}(V \setminus S), \vol_{W_2}(V \setminus S)) \geq \frac{\psi^2}{2} \min(\vol_{W_1 \cup W_2}(S), \vol_{W_1 \cup W_2}(V \setminus S))$.  Item 2, 3 are immediate.
\end{proof}

%    We say that $(W, \Pi)$ is a directed graph $G$, a $\phi$-out-expander with $\psi$-out-witness $(W, \Pi)$ if there exists a $\psi$-out-expander $W$ embedded into $G$ via $\Pi$ with congestion at most $\frac{\psi}{\phi}$ and
%    \begin{equation}\label{def:Witness}
%        \forall v \in V(W): \deg_{G}(v) \leq \deg_{W}(v) \leq \frac{\deg_G(v)}{\psi}.
%    \end{equation}
%    We call a directed graph $G$, a $\phi$-expander with $\psi$-witness $(W, \Pi) = (W_1 \cup W_2, \Pi_1 \cup \Pi_2)$ if $G$ is a $\phi$-out-expander with $\psi$-out-witness $(W_1, \Pi_1)$ and $G^{rev}$, where all edges are reversed, is a $\phi$-out-expander with $\psi$-out-witness $(W_2, \Pi_2)$.
\end{definition}

The next fact establishes that a $(\phi, \psi)$-(out-)witness for $G$ certifies that $G$ is a $\phi$-(out-)expander, justifying the name witness.

\begin{fact}[see \cite{hua2023maintaining}, Claim 2.1]\label{fact}
If $(W, \Pi)$ is a $(\phi,\psi)$-witness for $G$, then $G$ is a $\phi$-expander.
\end{fact}

\begin{proof}
    For any set $S \subseteq V(G)$, we have in the witness that $e_W(S, V \setminus S) \geq \psi \cdot \min(\vol_W(S) , \vol_W(V \setminus S)) \geq \psi \cdot \min(\vol_G(S) , \vol_G(V \setminus S))$. Since for any one of these cut edges there is an embedding path crossing the cut and since the congestion of the embedding is at most $\frac{\psi}{\phi}$, we have $e_G(S, V \setminus S) \geq \frac{\phi}{\psi} \cdot e_W(S, V \setminus S) \geq \phi \cdot \min(\vol_G(S) , \vol_G(V \setminus S)).$
\end{proof}

\begin{definition}[Augmented Expander Decomposition]\label{def:augmentedED}
    We call a collection $\mathcal{X}$ and a subset $E^r \subseteq E$ a $(\beta, \phi, \psi)$-expander decomposition of a graph $G$, if 
    \begin{enumerate}
        \item \label{Def:ED-item1} $\forall (X, W, \Pi) \in \mathcal{X}$, $G[X]$ has a $(\phi, \psi)$-witness $(W, \Pi)$,
        \item \label{Def:ED-item2} $|E^r| \leq \beta \cdot \phi \cdot e(G)$,
        \item \label{Def:ED-item3} $(G \setminus E^r)/ \mathcal{P}$ is a DAG, where $\mathcal{P} = \{X \mid (X, W, \Pi) \in \mathcal{X}\}$.
    \end{enumerate}
    Given two expander decompositions $(\mathcal{X}_1,E^r_1)$ of the graph $G$ and  $ (\mathcal{X}_2,E^r_2)$ of the graph $G \setminus \mathcal{D}$, where $\mathcal{D} \subseteq E(G)$, we say $(\mathcal{X}_2,E^r_2)$ refines $(\mathcal{X}_1,E^r_1)$ if 1) for all partition classes $X_2$, where $(X_2, W_2, \Pi_2) \in \mathcal{X}_2$, there is a class $X_1$, where 
    $(X_1, W_1, \Pi_1) \in \mathcal{X}_1$, such that $X_2 \subseteq X_1$ and 2) $E_1^r \subseteq E_2^r \cup \mathcal{D}$.
\end{definition}

\section{The Push-Pull-Relabel Framework}
\label{sec:pushpullRelabel}

In the push-relabel framework as presented in \cite{goldberg1988new}, we are trying to compute a feasible flow for a flow problem $(G, \cc, \Delta, \nabla)$ by maintaining a pre-flow $\ff$ together with a level function $\bell$. The algorithm then runs in iterations terminating once $\ff$ has no positive excess at any vertex. In each iteration of the algorithm, the algorithm identifies a vertex $v$ that still has positive excess at a vertex $v$. This positive excess is then pushed to neighbors on lower levels such that the capacity constraint is still enforced. If this is not possible $v$ is relabeled meaning that its label $\bell(v)$ is increased. In this section, we are extending the push-relabel framework to the dynamic setting, where we allow for increasing the source function $\Delta$ and the deletion of edges from $G$. To do so, we need to introduce the notion of negative excess. Because for a flow $\ff$ of the flow problem, it might happen that once we delete an edge $(u,v)$ there is more flow leaving the vertex $v$ than is entering or sourced, i.e. $\Gamma_t(v) < 0$. Hence, we need to extend the discussion to include negative excess in the dynamic version. Similar to the standard push-relabel algorithm, we maintain a pseudo-flow $\ff$ and a vertex labeling $\bell$ in so-called valid states $(\ff, \bell)$. The only difference is that in the standard push-relabel algorithm, $\ff$ is a pre-flow (not only a pseudo-flow). 

\begin{definition}\label{def:state}
    Given a level function $\bell: V(G) \rightarrow [h]$ and a pseudo-flow $\ff$, we call a tuple $(\ff, \bell)$ a state for $(G, \cc, \Delta, \nabla)$ if for all edges $e = (u, v)$ having $\bell(u) > \bell(v) + 1$ implies $\ff(e) = \cc(e).$\\
    We call the state $(\ff, \bell)$ valid if for all vertices $v \in V(G):$ 1) $\Gamma(v) < \nabla(v)/2$ implies $\bell(v) = 0$, 2) $\nabla(v) < \Gamma(v)$ implies $\bell(v) = h.$
\end{definition}

For the remainder of the paper, we have $h = O\left(\frac{\log(n)}{\phi}\right)$. To provide the reader with some intuition about the definitions, we remark that the pseudo-flow of a valid state is not feasible in the usual sense. Indeed, some vertices might have positive or negative excess. But these vertices are guaranteed to either be at level $h$ or at level $0$. Moreover, we point out that if a valid state $(\ff, \bell)$ has no vertices at level $h$, then $\ff$ might still not be a feasible flow since there might be a vertex $v$ with $\Gamma_t(v) < 0$. But it is straightforward to obtain from $\ff$ a feasible flow: extract from $\ff$ at each vertex $v$ exactly $\Delta(v)$ unit flow paths (possibly empty paths starting and ending in $v$). Let us briefly outline how we will use \Cref{lm:PushPullRelabel} in the directed expander pruning algorithm \ref{alg:DirectedExpanderPruning}. In the directed expander pruning algorithm, we start with an out-expander and an adversary performs a number of vertex or edge deletions. We aim to find a small pruning set of vertices such that if we prune away this small pruning set from the remaining graph, then we can certify that the obtained graph is still an expander. The classical way (see \cite{saranurak2019expander}) to either certify expansion of the remaining graph or to find a pruning set is by setting up a flow problem, where we inject for any deletion some additional source flow and where each vertex has sink capacity equal to its degree (this is the reason for the condition $\nabla \geq \deg $ in \Cref{lm:PushPullRelabel}). We will solve this flow problem using the flow provided by our \textsc{ValidState} algorithm (see \Cref{alg:ValidState}). We maintain this flow problem using the interface functions \textsc{IncreaseSource}, \textsc{RemoveEdges} under both adversarial deletions and under necessary pruning deletions. If the valid state computed after any such update has an additional vertex on level $h$, this will indicate that more vertices need to be pruned away. If on the other hand all vertices are on levels strictly below $h$, we will certify that the remaining graph is an expander.

\begin{lemma}\label{lm:PushPullRelabel}
Given a flow problem $(G = (V,E), \cc, \Delta, \nabla)$, where $\nabla \geq \deg$. Then, there is a deterministic data structure $\textsc{ValidState}(G, \cc, \Delta, \nabla)$ (see \Cref{alg:ValidState}) that initially computes a valid state $(\ff, \bell)$ and after every update of the form
\begin{itemize}
    \item $\textsc{IncreaseSource}(\mathbf{\delta}):$ where $\mathbf{\delta} \in \mathbb{N}_{\geq 0}^{n}$, we set $\Delta$ to $\Delta + \mathbf{\delta}$,
    \item $\textsc{RemoveEdges}(D)$: where $D \subseteq E(V) \setminus \mathcal{D}$, sets $\mathcal{D}$ to $\mathcal{D} \cup D$ (initially $\mathcal{D} = \emptyset$),
\end{itemize}
the algorithm explicitly updates the tuple $(\ff, \bell)$ such that thereafter $(\ff, \bell)$ is a valid state for the current flow instance $(G \setminus \mathcal{D}, \cc|_{E \setminus \mathcal{D}}, \Delta|_{E \setminus \mathcal{D}}, \nabla|_{E \setminus \mathcal{D}})$. The run-time is $O\left(h \cdot \left(\|\Delta\|_1 + \sum_{d \in \mathcal{D}} (1 + |\ff_{t_d}(d)|)\right)\right)$, where $\ff_{t_d}(d)$ is equal to $\ff(d)$ at the time $d$ is deleted and $\Delta$ is the variable at the end of the algorithm.
\end{lemma}

The user interface of the data structure $\textsc{ValidState}$ are the functions $\textsc{Init}, \textsc{IncreaseSource}$ and $ \textsc{RemoveEdges}$. These functions can be used to initialize or update the flow problem. After any such update the internal state $(\ff,\bell)$ has to be updated to remain valid for the new flow problem. To facilitate these necessary updates to $(\ff,\bell)$ these functions make calls to the internal functions $\textsc{PushRelabel}$ and $ \textsc{PullRelabel}$. As described before, the function $\textsc{PushRelabel}$ performs push and relabel operations to handle any positive excess, while the function $\textsc{PullRelabel}$ performs pull and relabel operations to handle any negative excess. Let us take a closer look at the individual functions: using $\textsc{Init}$ we set up the flow problem. It makes an internal call to $\textsc{PushRelabel}$ to compute the first valid state using the standard push-relabel algorithm. Note that at this point no negative excess has been introduced and thus we do not need to resort to any pull operations. The function $\textsc{IncreaseSource}$ can be used to increase the source of the current flow problem and again this does not introduce any negative excess and hence we can again update the valid state using the standard push-relabel algorithm. Using the function $\textsc{RemoveEdges}$, we can restrict the flow problem to a subgraph. The function induces the capacities $\cc, \Delta, \nabla$ and the flow $\ff$ to the subgraph and then makes calls to both $\textsc{PushRelabel}$ and $\textsc{PullRelabel}$ to update the state. The function $\textsc{PushRelabel}$ is just the standard push-relabel algorithm. We look for a vertex $v$ with positive excess $\Gamma(v) > \nabla(v)$ at level below $h$. We pick a vertex on the lowest level possible. We try to push flow along some unsaturated outgoing edge to a vertex on a lower level. If it is not possible to push flow, we increase the label of the vertex and repeat. The function $\textsc{PullRelabel}$ is the analog of the push-relabel algorithm for negative excess. Here, we look for a vertex $v$ with negative excess $\Gamma(v) < \nabla(v)/2$ at level above $0$. We pick a vertex on the highest level possible. We try to pull flow along some unsaturated incoming edge from a vertex on a higher level. If it is not possible to pull flow then all edges $(u,v)$ where $\bell(u) > \bell(v)$ must be saturated and it is safe to decrease the label of the vertex and repeat. 

%\aurelio{explain $f=\min \left(|\Gamma(v)|, \cc_f(u, v), \nabla(v) + \Gamma(u) \right)$ }

\begin{proof}[Proof of \Cref{lm:PushPullRelabel}.]
Since we only decrease the node label $\bell(v)$ when all edges $(u,v)$, where $\bell(u) > \bell(v)$, are saturated and since we only increase the node label when all edges $(v,w)$, where $\bell(v) > \bell(w)$, are saturated, we maintain that $(\ff,\bell$ is a state. Since we keep on performing push, pull or relabel operations as long as there is a vertex $v$ with $\Gamma(v) > \nabla(v), \bell(v) < h$ or $\Gamma(v) < \nabla(v)/2, \bell(v) > 0$, it is guaranteed that the algorithm computes a valid state eventually. It thus suffices to bound the run-time. \\
\\
Note that by maintaining at each vertex $v \in V$, a linked-list $L[v]$ containing all non-saturated edges $(v,u)$ where $\bell(v) = \bell(u) + 1$, we can implement a push in time $O(1)$. We can maintain such a linked-list for every vertex by spending time $O(\deg(v))$ every time we relabel a vertex $v$ plus $O(1)$ time for each push. We maintain a corresponding list $L'[v]$ for all pull operations. It thus suffices to bound the contribution of the push, pull, and relabel operations to the run-time. Let us index by $t$ starting with $0$ the push, pull, relabel operations as well as the calls to $\textsc{IncreaseSource}, \textsc{RemoveEdges}$ in the order they occur. We denote by $\tilde{E}_t, \ff_t, \bell_t, \Gamma_t$ the state of the variables at index $t$. We note that $\ff_t$ is supported on the edges $\tilde{E}_t$. \\
\\
Let us introduce functions $\textbf{p}_t, \textbf{n}_t$ where $\textbf{p}_t(v), \textbf{n}_t(v)$ can be interpreted as the amount of positive/ negative units present at vertex $v \in V$ at time $t$. These functions will satisfy 
\begin{equation}\label{eq:sum-of-protons-electrons}
    \forall v \in V, \Gamma_t(v) = \textbf{p}_t(v) - \textbf{n}_t(v).
\end{equation}
Initially, we have $\textbf{p}_0 = \Delta_0, \textbf{n}_0 = \boldsymbol{0},$ where $\Delta_0$ denotes the variable $\Delta$ at initialization of the data structure. We describe how $\textbf{p}_t, \textbf{n}_t$ evolve over time. If the t-th action is 
\begin{enumerate}
    \item a push of $\frac{1}{2}$ units from $u$ to $v$, then $\textbf{p}_{t+1} = \textbf{p}_t - \frac{1}{2} \cdot \mathbbm{1}_u + \frac{1}{2} \cdot \mathbbm{1}_v$
    \item a pull of $\frac{1}{2}$ units from $u$ to $v$, then $\textbf{n}_{t+1} = \textbf{n}_t + \frac{1}{2} \cdot \mathbbm{1}_u - \frac{1}{2} \cdot \mathbbm{1}_v$
    \item $\textsc{IncreaseSource}(\delta)$, then $\textbf{p}_{t+1} = \textbf{p}_t + \delta$
    \item $\textsc{RemoveEdges(D)}$, then 
    \begin{align*}
        \forall v \in V: \textbf{p}_{t+1}(v) = \textbf{p}_t(v) + \sum_{(v,w) \in D} \left(\ff_t\right)^+(v,w), \quad \textbf{n}_{t+1}(v) = \textbf{n}_t(v) + \sum_{(u,v) \in D} \left(\ff_{t}\right)^+(u,v)
    \end{align*}
\end{enumerate}
Clearly $\|\textbf{p}_t\|_1$ and $\|\textbf{n}_t\|_1$ are non-decreasing, we denote the final values by $\pi, \nu$. Let us verify that indeed the transitions above preserve the property in \eqref{eq:sum-of-protons-electrons}. For type 1/2/3 actions, this is immediate. For type 4, we observe that for any vertex $v \in V$ when inducing the flow $\ff_{t}$ onto $\tilde{E}_{t+1}$ every unit of $\sum_{(v,w) \in D} (\ff_t)^+(v,w)$ increases the commodity $\Gamma_t(v)$ at $v$ by one and every unit of $\sum_{(u,v) \in D} (\ff_t)^+(u,v)$ decreases the commodity $\Gamma_t(v)$ at $v$ by one.

\begin{claim}\label{proof:LmPushPullRelabel-cl1}
    We can bound the size of $\pi, \nu$ by
    \[O\left(\|\Delta\|_1 + \sum_{d \in \mathcal{D}} |\ff_{t_d}(d)| \right),\]
    where $\Delta$ denotes the final state of the variable.
\end{claim}

\begin{proof}
    Only \textsc{IncreaseSource} and \textsc{RemoveEdges} can increase $\|\textbf{p}_t\|_1, \|\textbf{n}_t\|_1$. The bound follows immediately from 3, 4 above.
\end{proof}

\begin{claim}\label{proof:LmPushPullRelabel-cl2} 
    The contribution to the run-time of the relabeling operations is bounded by 
\[O\left( \sum_{v \in V} \nabla(v) \sum_{0 \leq t} |\bell_{t+1}(v) - \bell_t(v)|\right) \leq O \left(h \cdot \pi \right).\]
\end{claim}

\begin{proof}
    Since $\deg(v) \leq \nabla(v)$, we can bound the run-time contribution of the relabeling operations by $O\left( \sum_{v \in V} \nabla(v) \sum_{0 \leq t} |\bell_{t+1}(v) - \bell_t(v)|\right)$. It thus suffices to bound this expression. We treat each $v \in V(G)$ individually. W.l.o.g. $\bell_t(v)$ increases at some time $t$, otherwise the contribution of $v$ to the sum is trivially zero. Let us consider the following time stamps $s_0 < s_1 < s_2 < \dots < s_l$ such that $s_0$ is the first time $\Gamma_t(v) \geq \nabla(v)$ and 
    \begin{align*}
        \forall i \geq 0: s_{2i + 1} &= \min\{t \geq s_{2i} \mid \Gamma_{t}(v) \leq \nabla(v)/2\}, \\
        \forall i \geq 1: s_{2i} &= \min\{t \geq s_{2i - 1} \mid \Gamma_t(v) \geq \nabla(v)\}.
    \end{align*}
    Additionally, we define the time stamps $t_0 < t_1 < t_2 < ... $ such that 
    \begin{align*}
        \forall i \geq 0: t_{2i} &= \max\{t < s_{2i + 1} \mid \Gamma_t(v) \geq \nabla(v)\}, \\
        \forall i \geq 1: t_{2i - 1} &= \max\{t < s_{2i} \mid \Gamma_t(v) \leq \nabla(v)/2\}.
    \end{align*}
    Let us introduce the vector $\textbf{z}_t(v) = \min(\textbf{p}_t(v), \textbf{n}_t(v) + \nabla(v)/2)$ (this corresponds to the neutrons introduced in the introduction). We observe that $\textbf{z}_t(v)$ is non-decreasing. Indeed, if we push away from $v$ there must be commodity excess, i.e. $\Gamma_t(v) \geq \nabla(v) + 1/2$ or equivalently $\textbf{p}_t(v) \geq \textbf{n}_t(v) + \nabla(v) + 1/2$, and if we pull towards $v$ there must be negative commodity excess, i.e. $\Gamma_t(v) \leq \nabla(v)/2 - 1/2$ or equivalently $\textbf{p}_t(v) \leq \textbf{n}_t(v) + \nabla(v)/2 - 1/2$. We will now establish that $\textbf{z}_{s_{2(i+1)}}(v) \geq \textbf{z}_{s_{2i}}(v) + \nabla(v)/2$. Consider the time stamps $t_{2i} \leq t < s_{2i+1}:$ since $\Gamma_{t+1}(v) < \nabla(v)$ the $t$-th operation cannot be a push away from $v$ and thus $\textbf{p}_{t_{2i}}(v) \leq \textbf{p}_{s_{2i+1}}(v)$. Combining it with the definition of the time stamp $s_{2i + 1}$, we find 
    \[\textbf{n}_{s_{2i + 1}}(v) + \nabla(v)/2 \geq \textbf{p}_{s_{2i + 1}}(v) \geq \textbf{p}_{t_{2i}}(v).\]
    This in turn implies that 
    \[\textbf{z}_{s_{2(i + 1)}}(v) \geq \textbf{z}_{s_{2i + 1}}(v) \geq \textbf{p}_{t_{2i}}(v) \geq \textbf{z}_{t_{2i}}(v) + \nabla(v)/2 \geq \textbf{z}_{s_{2i}}(v) + \nabla(v)/2,\]
    where we used in the first and the last inequality that $\textbf{z}_t(v)$ is non-decreasing and in the third the definition of the time stamp $t_{2i}$. By induction, we have $\textbf{z}_{s_{2i}}(v) \geq \textbf{z}_{s_{0}}(v) + i \cdot \nabla(v)/2 \geq \textbf{z}_{s_{0}}(v) + (i + 1) \cdot \nabla(v)/2 \geq (i + 1) \cdot \nabla(v)/2.$ This allows us to bound
    \begin{align*}
        \nabla(v) \cdot \sum_{s_0 \leq t} |\bell_{t+1}(v) - \bell_{t}(v)| &\leq \nabla(v) \cdot (l + 1) \cdot 2 \cdot h \leq 2 \cdot h \cdot (2 \lfloor l/2 \rfloor + 2) \cdot \nabla(v) \leq 4 \cdot h \cdot \textbf{z}_{s_{2\lfloor l/2 \rfloor}}(v), 
    \end{align*}
    where the $h$ is due to the fact that $\bell_k(v)$ is non-decreasing between $s_{2i}, s_{2i + 1}$ and non-increasing between $s_{2i + 1}, s_{2(i+1)}$. Summing over $v$ and bounding $\textbf{z}_t(v) \leq \textbf{p}_t(v)$ concludes the argument.
\end{proof}
\begin{claim}\label{proof:LmPushPullRelabel-cl3} 
    The contribution to the run-time of the push operations and the pull operations is bounded by 
    \[O\left(h \cdot \left( \|\Delta\|_1 + \sum_{d \in \mathcal{D}} |\ff_{t_d}(d)| \right) \right).\]
\end{claim}

\begin{proof}
    Let us introduce the function $\Phi(t) = \sum_{v \in V} \Gamma_t(v) \cdot \bell_t(v)$. We remark that any push or pull operation of $\frac{1}{2}$ units from $u$ to $v$, where $\bell_t(u) > \bell_t(v)$, decreases $\Gamma_t(u)$ by $\frac{1}{2}$ and increases $\Gamma_t(v)$ by at most $\frac{1}{2}$ while preserving all other $\Gamma_t$ entries. Hence, any such push or pull operation decreases the function $\Phi(t)$. This allows us to bound the run-time contribution of the push and pull operations by $\sum_{0 \leq t} \max(\Phi(t) - \Phi(t + 1), 0).$ Since at termination the state is valid and thus all vertices with negative excess are on level $0$, we have $\lim_t \Phi(t) \geq 0$ for all $t$. Together with $\Phi(0) = 0$, this implies that it suffices to bound the increases to $\Phi$. As argued before push and pull operations can only decrease the potential. So let $\mathcal{T}$ be the subset of indices that do not correspond to a push or a pull. Let us denote by $t_1, t_2, \dots, t_k$ all indices at which we perform a call to $\textsc{IncreaseSource}$ or $\textsc{RemoveEdges}$. 
    For any $i \geq 0,$  we denote by $\mathcal{T}_{t_i,t_{i+1}}$ all time indices $t$ in $\mathcal{T}$ where $t_i \leq t < t_{i+1}$ (where $t_{k+1} = \infty$ for notational convenience). To ease notation, we write 
    \[R = h \cdot \left( \|\Delta\|_1 + \sum_{d \in \mathcal{D}} |\ff_{t_d}(d)| \right).\]
    As a first step, we observe 
    \begin{align*}
        \sum_{t \in \mathcal{T}_{t_i,t_{i+1}}} \max(\Phi(t+1) - \Phi(t), 0) &= \sum_{t \in \mathcal{T}_{t_i + 1,t_{i+1}}} \sum_{v \in V} |\Gamma_t(v)| \cdot |\bell_{t+1}(v) - \bell_t(v)| \\
        &+ \sum_{v \in V}(\Gamma_{t_i + 1}(v) - \Gamma_{t_i}(v)) \cdot \bell_{t_i + 1}(v) ,
    \end{align*}
    where the first term of the expression is due to the relabeling operations in $\mathcal{T}_{t_i,t_i + 1}$ and the second term is due to the function call at index $t_i$. Since $\Gamma_{t_i}(v)$ can only increase if $\textbf{p}_{t_i}(v)$ increases, we may thus bound the contribution of a vertex of the second term by
    \begin{equation}\label{claim3-eq3}
        \left(\textbf{p}_{t_i + 1}(v) - \textbf{p}_{t_i}(v)\right) \cdot h.
    \end{equation}
    Combining \cref{proof:LmPushPullRelabel-cl1}, \cref{proof:LmPushPullRelabel-cl2}, we find that $\sum_i \sum_{t \in \mathcal{T}_{t_i + 1,t_{i+1}}} \sum_{v \in V} 2\nabla(v) \cdot |\bell_{t+1}(v) - \bell_t(v)| \leq O(R)$. We may thus subtract $\sum_{t \in \mathcal{T}_{t_i + 1,t_{i+1}}} \sum_{v \in V} 2\nabla(v) \cdot |\bell_{t+1}(v) - \bell_t(v)|$ and reduce the discussion to bounding the expression
    \begin{align}\label{eq6-PushPullRelabel}
        \sum_{t \in \mathcal{T}_{t_i + 1,t_{i+1}}} \sum_{v \in V} (|\Gamma_t(v)| - 2\nabla(v)) \cdot |\bell_{t+1}(v) - \bell_t(v)| &+ \sum_{v \in V}(\Gamma_{t_i + 1}(v) - \Gamma_{t_i}(v)) \cdot \bell_{t_i + 1}(v).
    \end{align}
    We treat each vertex $v \in V$ individually. Let $\tau_i(v) \geq 0$ be maximum such that $\forall t_i \leq t < t_i + \tau_i(v): \Gamma_t(v) > 2 \cdot \nabla(v)$ or $\Gamma_t(v) < -2 \cdot \nabla(v).$ Let us consider the first term of the expression \eqref{eq6-PushPullRelabel}. Note that if $|\Gamma_t(v)| \leq 2\nabla(v)$ then after any push or pull operation at time $t$, we still have that $|\Gamma_{t+1}(v)| \leq 2\nabla(v)$. This implies that for all $t \in \mathcal{T}_{t_i + \tau_i, t_{i+1}}: |\Gamma_t(v)| \leq 2\nabla(v).$  It thus suffices to consider for each $v \in V$ the expression %\aurelio{make this more prominent}
    
    \begin{align}\label{eq5-PushPullRelabel}
        \sum_{t \in \mathcal{T}_{t_i + 1,t_{i} + \tau_i(v)}} \left(|\Gamma_t(v)| - 2\nabla(v)\right) \cdot |\bell_{t+1}(v) - \bell_t(v)| &\leq \left(|\Gamma_{t_i + 1}(v)| - 2\nabla(v)\right) \cdot \sum_{t \in \mathcal{T}_{t_i + 1,t_{i}  + \tau_i(v)}} |\bell_{t+1}(v) - \bell_t(v)|.
    \end{align}
    
    Let us consider the vertices $v$ such that the quantity above is strictly positive. Clearly, $\Gamma_{t_i+1}(v) > 2\nabla(v)$ or $\Gamma_{t_i+1}(v) < -2\nabla(v)$. Let us consider the first case.\\ Since at time $t_i$ the tuple $(\ff_{t_i},\bell_{t_i})$ is valid, we must have that $\Gamma_{t_i}(v) < 2\nabla(v)$ otherwise $\bell_{t_i}(v) = h$ and hence $\sum_{t \in \mathcal{T}_{t_i + 1,t_{i}  + \tau_i(v)}} |\bell_{t+1}(v) - \bell_t(v)| = 0$. 
    Since $\Gamma_{t_i}(v)$ can only increase if $\textbf{p}_{t_i}(v)$ increases, we can bound 
    \[\left(|\Gamma_{t_i + 1}(v)| - 2\nabla(v)\right) \leq \left( \Gamma_{t_{i}+1}(v) - \Gamma_{t_i}(v) \right) \leq \textbf{n}_{t_i+1}(v) - \textbf{n}_{t_i}(v). \]
    Moreover, the label can only increase during $\mathcal{T}_{t_i + 1, t_i + \tau_i(v)}$ and hence $\sum_{t \in \mathcal{T}_{t_i + 1,t_{i}  + \tau_i(v)}} |\bell_{t+1}(v) - \bell_t(v)| \leq h$. We may thus bound the contribution of such a vertex to \eqref{eq5-PushPullRelabel} by
    \begin{equation}\label{claim3-eq1}
        \left(\textbf{p}_{t_i + 1}(v) - \textbf{p}_{t_i}(v)\right) \cdot h.
    \end{equation}
    If on the other hand $\Gamma_{t_i+1}(v) < -2\nabla(v)$, then since at time $t_i$ the tuple $(\ff_{t_i},\bell_{t_i})$ is valid we must have that $\Gamma_{t_i}(v) \geq 0$ otherwise $\bell_{t_i}(v) = 0$ and hence $\sum_{t \in \mathcal{T}_{t_i + 1,t_{i}  + \tau_i(v)}} |\bell_{t+1}(v) - \bell_t(v)| = 0$. Since $\Gamma_{t_i}(v)$ can only decrease if $\textbf{n}_{t_i}(v)$ increases, we can bound 
    \[\left(|\Gamma_{t_i + 1}(v)| - 2\nabla(v)\right) \leq \left( \Gamma_{t_i}(v) - \Gamma_{t_i + 1}(v) \right) \leq \textbf{n}_{t_i+1}(v) - \textbf{n}_{t_i}(v). \]
    Moreover, the label can only decrease during $\mathcal{T}_{t_i + 1, t_i + \tau_i(v)}$ and hence $\sum_{t \in \mathcal{T}_{t_i + 1,t_{i}  + \tau_i(v)}} |\bell_{t+1}(v) - \bell_t(v)| \leq h$. We may thus bound the contribution of such a vertex to \eqref{eq5-PushPullRelabel} by
    \begin{equation}\label{claim3-eq2}
        \left(\textbf{n}_{t_i + 1}(v) - \textbf{n}_{t_i + 1}(v)\right) \cdot h.
    \end{equation}
    Summing over all time indices $t_1, t_2, \dots$ and combining \cref{proof:LmPushPullRelabel-cl1}, \cref{proof:LmPushPullRelabel-cl2}, we can bound both equation \ref{claim3-eq3}, \ref{claim3-eq1}, \ref{claim3-eq2} by $O(R)$. This yields the conclusion. 
    \end{proof}
\end{proof}

\section{Directed Expander Decompositions via the Push-Pull-Relabel Framework}
\label{sec:dirExpanderDecomp}

In this section, we discuss our main technical result that states that if $G$ is initially a $\phi$-expander, then an expander decomposition can be maintained efficiently. The algorithm behind this theorem heavily relies on the new push-pull-relabel algorithm presented in the previous section.

\begin{theorem}\label{thm:ExpanderPruning} Given a $\phi$-expander $G = (V, E)$ with $(\phi,\psi)$-witness $(W, \Pi)$, there is a deterministic data structure $\textsc{BidirectedExpanderPruning}(G, W, \Pi)$ (see \Cref{alg:BiDirectedExpanderPruning}). After every update of the form
\begin{itemize}
    \item $\textsc{RemoveEdges}(D)$: where $D \subseteq E$, sets $\mathcal{D}$ to $\mathcal{D} \cup D$ (initially $\mathcal{D} = \emptyset$),
\end{itemize}
the algorithm explicitly updates $\tilde{V} \subseteq V$ and the tuple $(E^r, \mathcal{P})$, where $E^r \subseteq E(G)$ and $\mathcal{P}$ is a partition of $V \setminus \tilde{V}$ (initially $E^r , \mathcal{P} = \emptyset$), such that both $E^r$ and $\mathcal{P}$ are non-decreasing and such that after the update 
\begin{enumerate}
    \item \label{itm:DEP-thm1} the graph $\tilde{G} = \left(G \setminus \mathcal{D}\right)\left[\tilde{V} \right]$ has a $\left(\frac{\phi \cdot \psi^6}{32000}, \frac{\psi^4}{800}\right)$-witness $(\tilde{W}, \Pi)$,
    \item \label{itm:DEP-thm2} $|E^r| \leq \frac{\phi}{4} \cdot \sum_{P \in \mathcal{P}} \operatorname{vol}_{G}(P)$
    \item \label{itm:DEP-thm3} $\sum_{P \in \mathcal{P}} \operatorname{vol}_{W}(P) \leq \frac{8}{\psi} \cdot |\Pi^{-1}(\mathcal{D})|$
    \item \label{itm:DEP-thm4} $(G \setminus (\mathcal{D} \cup E^r))/ (\mathcal{P} \cup \{\tilde{V}\})$ is a directed acyclic graph (DAG),
\end{enumerate}
provided $\frac{200}{\psi^2} \cdot |\Pi^{-1}(\mathcal{D})| < e(G)$. The run-time is $O\left(\frac{h}{\psi^2} \cdot |\Pi^{-1}(\mathcal{D})| + h \cdot \sum_{d \in \mathcal{D} \cup E_{\mathcal{P}}} (1 + |\ff_{t_d}(d)|) \right) = O\left(h \cdot \frac{|\mathcal{D}|}{\psi^2 \phi} \right)$, where $\mathcal{D}$ denotes the variable at the end of all updates and $E_{\mathcal{P}}$ denotes all intercluster edges of the partition $\mathcal{P} \cup \{\tilde{V}\}$.
\end{theorem}

Following the high-level approach of \cite{saranurak2019expander}, we obtain our main result \Cref{Main-thm} for the special case where $G$ is static by running a generalization of the cut-matching game to directed graphs \cite{khandekar2009graph, louis2010cut} and then apply \Cref{thm:ExpanderPruning} to handle unbalanced cuts. Combining this result again with \Cref{thm:ExpanderPruning}, we obtain our main result \Cref{Main-thm} in full generality. Both of these reductions have been known from previous work \cite{bernstein2020deterministic, hua2023maintaining}. We defer the former reduction to \Cref{subsec:staticExpDecom} and the latter to \Cref{subsec:dynExpanderDecomposition}. 

\subsection{Reduction to Out-Expanders}

In this subsection, we show that the task of maintaining an expander decomposition as described in \Cref{thm:ExpanderPruning} can be reduced to a simpler problem that only requires maintaining out-expanders (however under edge and vertex deletions). In particular, we reduce to the following statement whose proof is deferred to \Cref{subsec:outExpan}.

\begin{lemma}\label{lm:DirectedExpanderDecomp}For every $\phi$-out-expander $G = (V, E)$ with $(\phi,\psi)$-witness $(W, \Pi)$, there is a deterministic data structure $\textsc{DirectedExpanderPruning}(G, W, \Pi)$ (see \Cref{alg:DirectedExpanderPruning}). After every update of the form (initially $\tilde{V} = V, \mathcal{S} = \emptyset, E_{\mathcal{S}} = \emptyset, E_{\mathcal{S}}^+ = \emptyset, E_{\mathcal{S}}^- = \emptyset, \mathcal{D} = \emptyset$)
\begin{itemize}
    \item $\textsc{RemoveEdges}(D)$: where $D \subseteq E(\tilde{V}) \setminus \mathcal{D}$, sets $\mathcal{D}$ to $\mathcal{D} \cup D$,
    \item $\textsc{RemoveVertices}(S)$: where $S \subseteq \tilde{V}$, sets $\tilde{V}$ to $\tilde{V} \setminus S$, $\mathcal{S}$ to $\mathcal{S} \cup S$, $E_{\mathcal{S}}^+$ to $E_{\mathcal{S}}^+ \cup E_G(S, \tilde{V} \setminus S), E_{\mathcal{S}}^-$ to $E_{\mathcal{S}}^- \cup E_G(\tilde{V} \setminus S, S)$ and $E_{\mathcal{S}}$ to $E_{\mathcal{S}}^+ \cup E_{\mathcal{S}}^-$
\end{itemize}
the algorithm explicitly updates $\tilde{V} \subseteq V$ and the tuple $(E^r, \mathcal{P})$, where $E^r \subseteq E$ and $\mathcal{P}$ is a partition of $V \setminus \left(\mathcal{S} \cup \tilde{V}\right)$ (initially $E^r = \emptyset, \mathcal{P} = \emptyset$), such that both $E^r$ and $\mathcal{P}$ are non-decreasing and such that after the update 
\begin{enumerate}
    \item \label{itm:DEP-lm1} the graph $\left(G \setminus \mathcal{D}\right)\left[\tilde{V}\right]$ has a $\left(\frac{\phi \cdot \psi^4}{400}, \frac{\psi^2}{20}\right)$-out-witness $(\tilde{W}, \tilde{\Pi})$,
    \item \label{itm:DEP-lm2} $|E^r| \leq \frac{\phi}{4} \cdot \sum_{P \in \mathcal{P}} \operatorname{vol}_G(P),$
    \item \label{itm:DEP-lm3} $\sum_{P \in \mathcal{P}} \operatorname{vol}_W(P) \leq \frac{4}{3 \psi} \cdot \left|\Pi^{-1}(\mathcal{D} \cup E^-_{\mathcal{S}}) \right|, $
    \item \label{itm:DEP-lm4} $E^r = \bigcup_{P} E_{G \setminus (\mathcal{D} \cup E^-_{\mathcal{S}})}(P,V_P)$, where $V_P$ is equal to $\tilde{V}$ at the time $P$ is added to $\mathcal{P}$,
\end{enumerate}
provided $\frac{20}{\psi} \cdot  |\Pi^{-1}(\mathcal{D} \cup E_{\mathcal{S}} \cup E_{\mathcal{P}})| < e(G)$, where $\mathcal{D}$ denotes the variable at the end of all updates and $E_{\mathcal{P}}$ denotes all intercluster edges of the partition $\mathcal{P} \cup \{\tilde{V}\}$. The algorithm runs in total time $O\left(\frac{h \cdot |\Pi^{-1}(\mathcal{D} \cup E_{\mathcal{S}} \cup E_{\mathcal{P}})|}{\psi^2} + h \cdot \sum_{d \in \mathcal{D} \cup E_{\mathcal{S}} \cup E_{\mathcal{P}}} (1 + |\ff_{t_d}(d)|)\right)$.
\end{lemma}

We present $\textsc{BiDirectedExpanderPruning}$ (see Algorithm \ref{alg:BiDirectedExpanderPruning}), which implements this reduction. The $\textsc{Init}$ function initializes two out-expander pruning data structure. One for the graph $G$ with witness $(W, \Pi)$ with variables $\tilde{V}_1, \mathcal{P}_1, E^r_1$ and one for the graph $G^{rev}$ with witness $(W, \Pi)$ with variables $\tilde{V}_2, \mathcal{P}_2, E^r_2$. From the guarantees of the one directional \textsc{DirectedExpanderPruning} algorithm, we will have that after each update $G[\tilde{V}_1]$ and $G^{rev}[\tilde{V}_2]$ are out-expanders. To conclude that the $G[\tilde{V}]$ is in fact an expander in both directions, we will ensure using the function \textsc{AdjustPartition} that $\tilde{V}_1, \tilde{V}_2$ agree. If they do not agree, we will remove the vertices $\tilde{V}_1 \setminus \tilde{V}_2 $ from $\tilde{V}_1$ using the function \textsc{RemoveVertices} of the algorithm \textsc{DirectedExpanderPruning} (see \Cref{alg:DirectedExpanderPruning}). Using the function \textsc{RemoveEdges} of \textsc{BiDirectedExpanderPruning}, we can remove edges from $G[\tilde{V}]$. This is again accomplish by calling the analogous function \textsc{RemoveEdges} of \textsc{DirectedExpanderPruning} for both directions. Again we have to adjust the partition such that $\tilde{V}_1, \tilde{V}_2$ agree afterwards. The algorithm and the proof are defered to \Cref{subsec:Reduction-to-Out-Expanders}.

\bibliographystyle{alpha}
\bibliography{refs}

\newpage

\appendix 
\input{appendix}

\end{document}

%% file: defs.tex
\usepackage[usenames,dvipsnames,svgnames,table]{xcolor}
\usepackage{enumitem}
\usepackage{graphicx}
\usepackage{fancybox}
\usepackage{comment}
\usepackage{xcolor}
\usepackage{nameref}
\usepackage[bottom]{footmisc}

\definecolor{ForestGreen}{rgb}{0.1333,0.5451,0.1333}
\definecolor{DarkRed}{rgb}{0.8,0,0}
\definecolor{Red}{rgb}{1,0,0}
\usepackage[linktocpage=true,
pagebackref=true,colorlinks,
linkcolor=ForestGreen,citecolor=ForestGreen,
bookmarks,bookmarksopen,bookmarksnumbered]
{hyperref}

\usepackage{amsmath}
\usepackage{amssymb}
\usepackage{amsthm}

\usepackage{algorithm}
\usepackage{algpseudocode}

\usepackage{subfig}

\usepackage{thm-restate}

\usepackage{float}
\usepackage{cleveref}

\newtheorem{theorem}{Theorem}[section]

\newtheorem{lemma}[theorem]{Lemma}

\newtheorem{claim}[theorem]{Claim}

\newtheorem{fact}[theorem]{Fact}

\newtheorem{definition}[theorem]{Definition}
\newtheorem{remark}[theorem]{Remark}

\newtheorem*{theorem*}{Theorem}
\newtheorem*{corollary*}{Corollary}
\newtheorem*{conjecture*}{Conjecture}
\newtheorem*{lemma*}{Lemma}
\newtheorem*{thm*}{Theorem}
\newtheorem*{prop*}{Proposition}
\newtheorem*{obs*}{Observation}
\newtheorem*{definition*}{Definition}

\newtheorem*{remark*}{Remark}
\newtheorem*{rec*}{Recommendation}

\newenvironment{fminipage}%
  {\begin{Sbox}\begin{minipage}}%
  {\end{minipage}\end{Sbox}\fbox{\TheSbox}}

\newcommand\bell{\boldsymbol{\mathit{\ell}}}

\newcommand\cc{\boldsymbol{\mathit{c}}}

\newcommand\ff{\boldsymbol{\mathit{f}}}

\renewcommand\ss{\boldsymbol{\mathit{s}}}
\def\tt{\boldsymbol{\mathit{t}}}

\newcommand\xx{\boldsymbol{\mathit{x}}}

\newcommand\veczero{\boldsymbol{0}}
\newcommand\vecone{\boldsymbol{1}}

\newcommand\BB{\boldsymbol{\mathit{B}}}

\newcommand\R{\mathbb{R}}

\newcommand{\trp}{\top}

\algdef{SE}[SUBALG]{Indent}{EndIndent}{}{\algorithmicend\ }%
\algtext*{Indent}
\algtext*{EndIndent}

%% file: appendix.tex
\section{Appendix}

\input{prev_work}

\subsection{Push-Pull-Relabel Algorithm}
\label{subsec:Push-Pull-Relabel-Algorithm}

\begin{algorithm}[H]
    \begin{algorithmic}
    \caption{$\textsc{ValidState}(G = (V,E), \cc, \Delta, \nabla, h)$}\label{alg:ValidState}
    \State def $\textsc{Init}$
    \Indent
    \State $\tilde{E} \leftarrow E(G), (\cc,\Delta,\nabla) \leftarrow (\cc,\Delta,\nabla), (\ff, \bell) \leftarrow (\veczero,\veczero)$
    \State $\textsc{PushRelabel}()$
    \EndIndent
    \State 
    \State def $\textsc{IncreaseSource}(\delta)$
    \Indent
    \State $\Delta \leftarrow \Delta + \delta$
    \State $\textsc{PushRelabel}()$
    \EndIndent
    \State 
    \State def $\textsc{RemoveEdges}(D)$
    \Indent
    \State $\tilde{E} \leftarrow \tilde{E} \setminus D, (\ff, \bell) \leftarrow (\ff|_{\tilde{E}},\bell|_{\tilde{V}})$
    \State $\textsc{PullRelabel}()$
    \State $\textsc{PushRelabel}()$
    \EndIndent
    \State
    \State def $\textsc{PushRelabel}()$
    \Indent
    \While{$\exists v$ where $\bell(v)<h$ and $\Gamma(v) > \nabla(v)$}
        \State Let $v$ be a vertex minimizing $\bell(v)$.
        \If{$\exists (v, u)$ such that $\cc_f(v, u)>0, \bell(v)=\bell(u)+1$}
            %\State $f =\min \left(\Gamma(v) - \nabla(v), \cc_f(v, u), 2\nabla(u) - \Gamma(u) \right)$
            \State $\ff(v, u) \leftarrow \ff(v, u) + \frac{1}{2}, \ff(u, v) \leftarrow -\ff(v, u)$ // Sends $\frac{1}{2}$ units of pos. excess from $v$ to $u$
        \Else
            \State $\bell(v) \leftarrow \bell(v) + 1$
        \EndIf
    \EndWhile
    \EndIndent
    \State
    \State def $\textsc{PullRelabel}()$
    \Indent
    \While{$\exists v$ where $\bell(v) > 0$ and $\Gamma(v) < \nabla(v)/2$ }
        \State Let $v$ be a vertex maximizing $\bell(v)$
        \If{$\exists (u, v)$ such that $\cc_f(u, v)>0, \bell(u)=\bell(v)+1$ }
        %\State $f=\min \left(\Gamma(u) - \nabla(v)/2, \cc_f(u, v), \nabla(v) - \Gamma(v) \right)$ 
        \State $\ff(v, u) \leftarrow \ff(v, u) - \frac{1}{2}, \ff(u, v) \leftarrow - \ff(v, u)$ // Sends $\frac{1}{2}$ units of neg. excess from $v$ to $u$
        \Else
            \State $\bell(v) \leftarrow \bell(v) - 1$
        \EndIf
    \EndWhile
    \EndIndent
    \end{algorithmic}
\end{algorithm}

\subsection{Reduction to Out-Expanders}
\label{subsec:Reduction-to-Out-Expanders}

\begin{algorithm}[H]
    \begin{algorithmic}
    \caption{$\textsc{BiDirectedExpanderPruning}(G, W, \Pi)$}\label{alg:BiDirectedExpanderPruning}
    \State def $\textsc{Init}$
    \Indent
    \State $(\tilde{V}_1, \mathcal{P}_1, E^r_1) \leftarrow \textsc{DirectedExpanderPruning}(G, W, \Pi)$
    \State $(\tilde{V}_2, \mathcal{P}_2, E^r_2) \leftarrow \textsc{DirectedExpanderPruning}(G^{rev}, W, \Pi)$
    \State $\tilde{V} \leftarrow \tilde{V}_1 \cap \tilde{V}_2, \mathcal{P} \leftarrow \mathcal{P}_1 \cup \mathcal{P}_2, E^r \leftarrow E^r_1 \cup E^r_2$ // dynamically updated
    \EndIndent
    \State 
    \State def $\textsc{RemoveEdges}(D)$
    \Indent
    \State $(\tilde{V}_1, \mathcal{P}_1, E^r_1).\textsc{RemoveEdges}(D)$ 
    \State $(\tilde{V}_2, \mathcal{P}_2, E^r_2).\textsc{RemoveVertices}(\tilde{V}_2 \setminus \tilde{V}_1)$
    \State $(\tilde{V}_2, \mathcal{P}_2, E^r_2).\textsc{RemoveEdges}(D)$ 
    \State $\textsc{AdjustPartition}$
    \EndIndent
    \State
    \State def $\textsc{AdjustPartition}()$
    \Indent
    \While{$\tilde{V}_1 \neq \tilde{V}_2$}
    \State $(\tilde{V}_1, \mathcal{P}_1, E^r_1).\textsc{RemoveVertices}(\tilde{V}_1 \setminus \tilde{V}_2)$
    \State $(\tilde{V}_2, \mathcal{P}_2, E^r_2).\textsc{RemoveVertices}(\tilde{V}_2 \setminus \tilde{V}_1)$
    \EndWhile
    \EndIndent
    \end{algorithmic}
    \label{alg:bidirectedExp}
\end{algorithm}

\begin{proof}[Proof of \Cref{thm:ExpanderPruning}]
    Let $E^r, E^r_1, E^r_2, \mathcal{P}, \mathcal{P}_1, \mathcal{P}_2, \mathcal{D}$ denote the final values of the variables. To invoke the conclusions of \Cref{lm:DirectedExpanderDecomp}, we need to guarantee that the calls to both one directional data structures satisfy for $i \in \{1,2\}: \frac{20}{\psi} \cdot |\Pi^{-1}(\mathcal{D} \cup E_{\mathcal{S}_i} \cup E_{\mathcal{P}_i})| < e(G)$. Once we have established the assumption of \Cref{lm:DirectedExpanderDecomp}, \cref{itm:DEP-thm1}, \ref{itm:DEP-thm2}, and \ref{itm:DEP-thm4} will be direct consequences of \Cref{lm:DirectedExpanderDecomp}.\\
    \\
    We observe that $E_{\mathcal{S}_1}^- = E^r_2$ and $E_{\mathcal{S}_2}^- = E^r_1$. Combining it with \cref{itm:DEP-lm2} and \cref{itm:DEP-lm3} of Lemma \ref{lm:DirectedExpanderDecomp}, one obtains $|E^r_1| \leq \frac{\phi}{3\psi} \cdot \left(|\Pi^{-1}(E_2^r)| + |\Pi^{-1}(\mathcal{D})| \right).$ By symmetry, we also have $|E^r_2| \leq \frac{\phi}{3\psi} \cdot \left(|\Pi^{-1}(E_1^r)| + |\Pi^{-1}(\mathcal{D})| \right)$. Combining the two bounds yields
    \begin{align}\label{proof-biDEP}
        \frac{\phi}{\psi} \cdot |\Pi^{-1}(E^r)| &\leq |E^r| = |E^r_1| + |E^r_2| \leq \frac{\phi}{3 \psi} \cdot \left(|\Pi^{-1}(E^r)| + 2 \cdot |\Pi^{-1}(\mathcal{D})| \right),
    \end{align}
    where we used that the congestion of $\Pi$ is at most $\frac{\psi}{\phi}$ in the first inequality and that $E^r = E^r_1 \cup E^r_2$. We may thus bound $|\Pi^{-1}(E^r)| \leq 2 \cdot |\Pi^{-1}(\mathcal{D})|$.\\
    \\
    Using this fact together with $E_{\mathcal{S}_1}^-, E_{\mathcal{S}_2}^- \subseteq E_{\mathcal{S}}$, item \ref{itm:DEP-lm3} of Lemma \ref{lm:DirectedExpanderDecomp} implies item 3 directly. To bound $\Pi^{-1}(E_{\mathcal{P}})$, we note that any embedding path of $\Pi$ that enters some $P \in \mathcal{P}$ also leaves it or ends in $P$. Since there are at most $\vol_W(P)$ embedding paths ending in $P$, we can bound 
    \begin{align*}
        |\Pi^{-1}(E_{\mathcal{P}})| &\leq \sum_{P \in \mathcal{P}} \vol_W(P) + |\Pi^{-1}(E^r)| \leq \frac{4}{\psi} \cdot |\Pi^{-1}(\mathcal{D})| + 2 \cdot |\Pi^{-1}(\mathcal{D})| \\
        &\leq \frac{6}{\psi} \cdot |\Pi^{-1}(\mathcal{D})|,
    \end{align*}
    where we used in the second inequality item 3 and $|\Pi^{-1}(E^r)| \leq 2 \cdot |\Pi^{-1}(\mathcal{D})|$. Since $E_{\mathcal{S}_1} \cup E_{\mathcal{P}_1}, E_{\mathcal{S}_2} 
    \cup E_{\mathcal{P}_2} \subseteq E_{\mathcal{P}}$, this in particular establishes the assumption of \Cref{lm:DirectedExpanderDecomp}.\\
    \\
    For the run-time bound, we observe that $O\left(\frac{h}{\psi^2} \cdot |\Pi^{-1}(\mathcal{D} \cup E_{\mathcal{P}})| + h \cdot \sum_{d \in \mathcal{D} \cup E_{\mathcal{P}}} |\ff_{t_d}(d)| \right)$ follows immediately from the run-time bound of Lemma \ref{lm:DirectedExpanderDecomp} and the fact that $E_{\mathcal{S}_1}, E_{\mathcal{S}_2} \subseteq E_{\mathcal{P}}$. Using that $|\Pi^{-1}(E_{\mathcal{P}})| \leq \frac{6}{\psi} \cdot |\Pi^{-1}(\mathcal{D})|$ and the congestion of $\Pi$ is bounded $\frac{\psi}{\phi}$, the bound of the first term simplifies to $O\left(\frac{h}{\psi^3} \cdot |\Pi^{-1}(\mathcal{D})| \right) = O\left(\frac{h}{\psi^2 \phi} \cdot |\mathcal{D}| \right)$. To bound the second term of the run-time, we consider the edges in $E_{\mathcal{P}} = \bigcup_{P \in \mathcal{P}} E_{G \setminus \mathcal{D}_P}(P,V_P) \cup E(V_P, P)$, where $V_P$ is $\tilde{V}$ and $\mathcal{D}_P$ is equal to $\mathcal{D}$ at the time $t_P$ when $P$ is deleted. Since every unit of flow entering any $P \in \mathcal{P}$ must end up in a sink of $P$ or flow out again, we may bound $\sum_{d \in E_{P}} |\ff_{t_d}(d)| \leq \vol_W(P) + \sum_{d \in E^r \cap E_P} |\ff_{t_d}(d)|,$ where $E_P = E_{G \setminus \mathcal{D}_P}(P,V_P) \cup E(V_P, P)$. Summing over all $P$ yields 
    \begin{align*}
        \sum_{d \in E_{\mathcal{P}}} |\ff_{t_d}(d)| &\leq \sum_{P \in \mathcal{P}} \vol_W(P) + \sum_{d \in E^r} |\ff_{t_d}(d)| \leq \frac{4}{\psi} \cdot |\Pi^{-1}(\mathcal{D})| + \frac{16 \cdot E^r}{\phi \psi^2} \\
        &\leq \frac{20 \cdot \Pi^{-1}(\mathcal{D})}{\psi^3} \leq \frac{20 \cdot |\mathcal{D}|}{\psi^2 \phi},
    \end{align*}
    where we used in the second inequality item 3 and that the capacity is bounded by $\frac{16}{\phi \cdot \psi^2}$, in the third we combined item 2 and 3 and in the last we used the bound on the congestion of $\Pi$. 
\end{proof}

\subsection{Maintaining Out-Expanders}
\label{subsec:outExpan}

It remains to prove \Cref{lm:DirectedExpanderDecomp}.    \Cref{alg:DirectedExpanderPruning} gives an implementation of $\textsc{DirectedExpanderPruning}$. The $\textsc{Init}$ function initializes for an out-expander $(G, W, \Pi)$ the expander pruning variables $\tilde{V}, \mathcal{P}, E^r, \mathcal{D}$ as well as a first valid state $(\ff, \bell)$ of the $\textsc{ValidState}$ algorithm (see \Cref{alg:ValidState}). This valid state $(\ff, \bell)$, we will use to find the pruning cuts in function \textsc{Prune}. Using $\textsc{RemoveEdges}$ the user can remove edges $D$ from $G[\tilde{V}]$. Removing edges might force us to prune away some part of the $G[\tilde{V}]$ and add the pruned part to the collection of pruning sets $\mathcal{P}$. To find the pruning set $P$, we adopt a similar strategy as in \cite{saranurak2019expander}. We inject additional source flow into the flow problem: for any witness edge $e = (u,v) \in E(W)$, where an edge on the embedding path $\Pi(e)$ is in $D$, we increase the sources $\Delta(u), \Delta(v)$ by $\frac{4}{\psi}$. Since we updated the flow problem by increasing the source capacities and removing edges, we need to update the valid state $(\ff, \bell)$. This new valid state $(\ff, \bell)$, then allows us to locate the pruning set in the function \textsc{AdjustPartition}. In \textsc{AdjustPartition}, we check whether there is a vertex in the remaining graph $G[\tilde{V}]$ on level $h$. If there is no such vertex, it will certify that in fact $G[\tilde{V}]$ has already a good out-witness and we don't need to prune away any subgraph. If on the other hand, there is a vertex on level $h$ we will find a pruning set $P$ using the function \textsc{Prune}. This set $P$ is then added to the collection of pruning sets $\mathcal{P}$ and the vertices in $P$ are removed from $\tilde{V}$. The cut edges $E_{G\setminus \mathcal{D}}(P,\tilde{V})$ are added to the remove edges $E^r$. Since $\tilde{V}$ has become smaller, we need to update the valid state $(\ff, \bell)$ again. We do this once more by injecting additional source at the endpoints of any witness edge $e \in E(W)$, where some edge on the path $\Pi(e)$ is in $E_{G\setminus \mathcal{D}}(\tilde{V},P)$. Thereafter we again check whether in the new valid state $(\ff, \bell)$ there is a vertex in $\tilde{V}$ on level $h$. We keep on doing this procedure until there no longer is a vertex on level $h$. We will prove that the volume of the pruning sets can be related to the size of the set of edges $D$ initially removed by the user. So far we have not explained how to find the pruning set in function \textsc{Prune}. This is accomplished by a standard level cutting procedure. We start with $P$ being all vertices on level $h$ and then check whether the vertices on the next level have volume at least $\phi \cdot \vol_G(P)$. If it is the case, we add vertices on the next level to $P$ and otherwise return $P$.

Through the function \textsc{RemoveVertices} the user can remove vertices $S$ from $\tilde{V}$. Similar to the function \textsc{RemoveEdges}, we will have to inject additional source at the endpoints of the witness edges $e \in E(W)$, where some edge of $\Pi(e)$ is in $E_{G \setminus \mathcal{D}}(\tilde{V},S)$, and update the valid state $(\ff, \bell)$ accordingly. This might again leave some vertices on level $h$ and will again force us to prune some vertices. This is again accomplished by a call to \textsc{AdjustPartition}. And similar to \textsc{RemoveEdges}, we will again prove that the volume of the pruning sets can be related to the volume of the set $S$ initially removed by the user.

\begin{comment}
\begin{algorithm}[H]
\caption{$\textsc{PruneOrCertify}(G, \ell)$}\label{alg:PruneOrCertify}
\begin{algorithmic}[1]
\State $S \leftarrow \emptyset, i \leftarrow h$ 
\Repeat 
\State $S \leftarrow S \cup\left\{v \in \tilde{V} \mid \bell(v) = i\right\}$
\State $i \leftarrow i - 1$
\Until{$\operatorname{vol}_{G}\left(\{v \in \tilde{V} \mid \bell(v) \leq i\}\right) < (1 + \phi) \cdot \operatorname{vol}_{G}(S)$} 
\State \Return the cut $S$ 
\end{algorithmic}
\end{algorithm}
\end{comment}

\begin{algorithm}[H]
    \caption{$\textsc{DirectedExpanderPruning}(G, W, \Pi)$}\label{alg:DirectedExpanderPruning}
    \begin{algorithmic}[1]
    \State def $\textsc{Init}$ 
    \Indent
    \State $\tilde{V} \leftarrow V, \mathcal{P} \leftarrow \emptyset, E^+_{\mathcal{S}} \leftarrow \emptyset, E^-_{\mathcal{S}} \leftarrow \emptyset, E^r \leftarrow \emptyset, \mathcal{D} \leftarrow \emptyset$
    \State $(\ff, \bell) \leftarrow \textsc{ValidState}(G, \frac{18}{\phi \psi^2}\cdot \vecone,\veczero,\deg_W,h)$
    \EndIndent
    \State 
    \State def $\textsc{RemoveEdges}(D)$
    \Indent
    \State $\mathcal{D} \gets \mathcal{D} \cup D$
    \State $(\ff, \bell).\textsc{RemoveEdges}(D)$ 
    \State $(\ff, \bell).\textsc{IncreaseSource}(\frac{4}{\psi} \deg_{\Pi^{-1}(D)})$ 
    \State $\textsc{AdjustPartition}()$
    \EndIndent
    \State 
    \State def $\textsc{RemoveVertices}(S)$
    \Indent
    \State $\tilde{V} \leftarrow \tilde{V} \setminus S, E^+_{\mathcal{S}} \leftarrow E^+_{\mathcal{S}} \cup E_{G\setminus \mathcal{D}}(S,\tilde{V}), E^-_{\mathcal{S}} \leftarrow E^-_{\mathcal{S}} \cup E_{G\setminus \mathcal{D}}(\tilde{V},S)$
    \State $(\ff, \bell).\textsc{RemoveEdges}\left(E_{G\setminus \mathcal{D}}(\tilde{V},S) \cup E_{G\setminus \mathcal{D}}(S,\tilde{V})\right)$ 
    \State $(\ff, \bell).\textsc{IncreaseSource}(\frac{4}{\psi} \deg_{\Pi^{-1}(E_{G\setminus \mathcal{D}}(\tilde{V},S) \cup E_{G\setminus \mathcal{D}}(S,\tilde{V}))})$ 
    \State $\textsc{AdjustPartition}()$
    \EndIndent
    \State
    \State def $\textsc{AdjustPartition}()$
    \Indent
    \While{$\exists v \in \tilde{V}$ with $\bell(v) = h$}
    \State $P \leftarrow \textsc{Prune}(G \setminus \mathcal{D},\bell)$
    \State $\mathcal{P} \leftarrow \mathcal{P} \cup \{P\}, \tilde{V} \leftarrow \tilde{V} \setminus P, E^r \leftarrow E^r \cup E_{G \setminus \mathcal{D}}(P,\tilde{V})$ \label{alg:DiExpDecomp-AdjustPartition}
    \State $(\ff, \bell).\textsc{RemoveEdges}(E_{G \setminus \mathcal{D}}(P,\tilde{V}) \cup E_{G \setminus \mathcal{D}}(\tilde{V},P))$ 
    \State $(\ff, \bell).\textsc{IncreaseSource}(\frac{4}{\psi} \deg_{\Pi^{-1}\left(E_{G \setminus \mathcal{D}}(P,\tilde{V}) \cup E_{G \setminus \mathcal{D}}(\tilde{V},P)\right)})$
    \EndWhile
    \EndIndent
    \State
    \State def $\textsc{Prune}(G, \bell)$ 
    \Indent
    \State $S \leftarrow \emptyset, i \leftarrow h$ 
    \Repeat 
    \State $S \leftarrow S \cup\left\{v \in \tilde{V} \mid \bell(v) = i\right\}$
    \State $i \leftarrow i - 1$
    \Until{$\operatorname{vol}_{G \setminus \left(\mathcal{D} \cup E^-_{\mathcal{S}}\right)}\left(\{v \in \tilde{V} \mid \bell(v) \leq i\}\right) < (1 + \frac{\phi}{72}) \cdot \operatorname{vol}_{G \setminus \left(\mathcal{D} \cup E^-_{\mathcal{S}}\right)}(S)$} 
    \State \Return the cut $S$ 
    \EndIndent
    \end{algorithmic}
\end{algorithm}

\begin{proof}[Proof of Lemma \ref{lm:DirectedExpanderDecomp}]
Before we start with the actual proof, let us specify the sink and source capacity vectors $\Delta, \nabla \in \R^{V(G)}$. The sink capacities $\nabla(v)$ are fixed during the entire algorithm in line 3 to $\deg_W(v)$. The source capacities $\Delta$ are initially $\boldsymbol{0}$ (see again line 3), but during the course of the algorithm $\Delta$ is increased in lines 8, 14, 22. We remark that item \ref{itm:DEP-lm4} follows immediately from line \ref{alg:DiExpDecomp-AdjustPartition} in $\textsc{AdjustPartition}$. To establish item \ref{itm:DEP-lm1}, it suffices to prove the next claim.

\begin{claim}
    After every call to $\textsc{AdjustPartition}$, $(G \setminus \mathcal{D})[\tilde{V}]$ has a $\left(\frac{\phi \psi^4}{400}, \frac{\psi^2}{20}\right)$-out-witness $\tilde{W}$.
\end{claim}

\begin{proof}
In order to construct $(\tilde{W}, \tilde{\Pi})$ we first pick a path-decomposition of $\ff$ such that at every vertex exactly $\max(0, \ff(v))$ paths begin and at most $\max(0, -\ff(v))$ end. \\
\\
We initialize $\tilde{W}$ as $W_0 := \left(W \setminus \Pi^{-1}(\mathcal{D} \cup E_{\mathcal{S}} \cup E_{\mathcal{P}})\right)[\tilde{V}]$. For any vertex $v \in \tilde{V},$ we pick $\min(\max(0, \ff(v)) , \Delta(v))$ paths starting with $v$ (remark that $v$ might still have negative excess and thus $\ff(v) > \Delta(v)$). Note that any such path $p_{v,w}$ starting in $v$ and ending in $w$, is a path in $\tilde{V}$ because we removed all edges between any $P \in \mathcal{P}$ and $\tilde{V}$ and all edges between $\mathcal{S}$ and $\tilde{V}$. We add the edge $e = (v,w)$ to $\tilde{W}$ and embed this edge through $\tilde{\Pi}(e) = p_{v,w}$. If $\max(0, \ff(v)) < \Delta(v)$ then we additionally add $\Delta(v) - \max(0, \ff(v))$ self-loops to $v$. We observe that we added exactly $\Delta(v)$ out-edges to $v$. Since $(\ff, \bell)$ is valid by \Cref{lm:PushPullRelabel} and $\forall v \in \tilde{V}: \bell(v) < h$, we added at most $\Delta(v) - \ff(v) \leq \nabla(v)$ incoming edges to $v$. Hence, we can relate the degree in $\tilde{W}$ to the degree in $W$ and in $G$ by
\begin{align}\label{proof:PruneOrCertify1}
    \deg_{\tilde{W}}(v) &\leq \deg_{W_0}(v) + \nabla(v) + \Delta(v) \nonumber \leq \deg_{W_0}(v) + \deg_W(v) + \frac{4}{\psi}\left(\deg_W(v) - \deg_{W_0}(v)\right) \nonumber \\
    &\leq \frac{5}{\psi} \deg_W(v) \leq \frac{5}{\psi^2} \deg_G(v)
\end{align}
and by the upper-bound 
\begin{align*}
    \deg_G(v) &\leq \deg_W(v) = \deg_{W_0}(v) + \left(\deg_W(v) - \deg_{W_0}(v)\right) \\
    &\leq \deg_{W_0}(v) + \frac{4}{\psi}\left(\deg_W(v) - \deg_{W_0}(v)\right) \leq \deg_{W_0}(v) + \Delta(v) \leq \deg_{\tilde{W}}(v),
\end{align*}
where the first inequality is due to \eqref{def:Witness}, and the last due to the fact that we add at least $\Delta(v)$ out-going edges to $v$ when constructing $\tilde{W}$ from $W$. We point out that the new embedding $\tilde{\Pi}$ is a combination of the old embedding $\Pi$ and flow paths with congestion at most $\frac{18}{\phi \psi^2}$. Hence, it is clear that $\tilde{\Pi}$ is an embedding with congestion at most $\frac{20}{\phi\psi^2}$. What remains to be proven is that $\tilde{W}$ is a $\frac{\psi^2}{20}$-out-expander.\\
\\
Let us consider an arbitrary cut $(S, \tilde{V} \setminus S)$ where $\vol_{\tilde{W}}(S) \leq e(\tilde{W})$. To bound $e_{\tilde{W}}(S, \tilde{V} \setminus S)$, we observe that 
\[e_{\tilde{W}}(S, \tilde{V} \setminus S) \geq e_{W_0}(S,\tilde{V} \setminus S) + \left|\tilde{\Pi}^{-1}\left(E_{G \setminus (\mathcal{D} \cup E_{\mathcal{S}}^-)}(S, \tilde{V} \setminus S)\right)\right|. \]
We observe that if $e_{W_0}(S,\tilde{V}\setminus S) \geq \frac{1}{2} e_{W}(S,V \setminus S)$ then we immediately find that the cut is dense enough, because $e_{W}(S,V\setminus S) \geq \psi \cdot \min \left(\vol_W(S), \vol_W(V \setminus S)\right) \geq \frac{\psi^2}{5} \min \left(\vol_{\tilde{W}}(S), \vol_{\tilde{W}}(\tilde{V} \setminus S)\right)$. We may thus assume the contrary. As a second observation, 
\begin{align*}
    \vol_W(S) &\leq \vol_{\tilde{W}}(S) \leq e(\tilde{W}) \leq e(W) + \frac{4}{\psi} \cdot |\Pi^{-1}(\mathcal{D} \cup E_{\mathcal{S}} \cup E_{\mathcal{P}})| \\
    &\leq \left(1 + \frac{1}{5}\right) \cdot e(W),
\end{align*}
where the last inequality is true by the assumptions of the Lemma. This in particular implies that $\vol_W(S) \leq \frac{6}{5} \cdot \frac{5}{4} \cdot \left(\frac{4}{5} \cdot e(W) \right) \leq \frac{3}{2} \cdot \vol_W(V \setminus S)$. We note that exactly 
\begin{align*}
    \sum_{v \in S} \Delta(v) &\geq \frac{4}{\psi} (e_W(S, V \setminus S) - e_{W_0}(S, \tilde{V} \setminus S)) \\
    &\geq \frac{2}{\psi} e_W(S, V \setminus S) \geq 2 \cdot \min\left(\vol_W(S), \vol_W(V \setminus S)\right) \geq \frac{4}{3} \vol_W(S)
\end{align*}
embedding paths start in $S$ and at most $\sum_{v \in S} \nabla(v) \leq \vol_W(S)$ end in $S$. Therefore at least $\frac{1}{3} \cdot \vol_W(S) \geq \frac{\psi}{15} \vol_{\tilde{W}}(S)$ edges cross the cut $(S,\tilde{V} \setminus S)$.

\end{proof}

It remains to prove items \ref{itm:DEP-lm2} and \ref{itm:DEP-lm3}. Item \ref{itm:DEP-lm2} is a direct consequence of the next claim.

\begin{claim}
    At the end of each while-loop iteration of $\textsc{AdjustPartition}()$, we have that $e_{G \setminus (\mathcal{D} \cup E_{\mathcal{S}}^-)}(S, \tilde{V}) \leq \frac{\phi}{4} \cdot \operatorname{vol}_{G}(S),$ where we denote by $S$ the set as defined in the current iteration in Line 18.
\end{claim}

\begin{proof}
Let us consider the sets $S_i = \{v \in \tilde{V} \mid \bell(v) \geq i\}.$ We observe that $\operatorname{vol}_G(S_1) \leq 2m$. Thus, by our choice of $h$, there exists $1 \leq i < h$ such that $\operatorname{vol}_{G \setminus \left(\mathcal{D} \cup E^-_{\mathcal{S}}\right)}(S_{i}) < (1 + \frac{\phi}{72}) \cdot \operatorname{vol}_{G \setminus \left(\mathcal{D} \cup E^-_{\mathcal{S}}\right)}(S_{i+1})$ and thus the loop of $\textsc{Prune}$ breaks. We observe that
\begin{align*}
    e_{G \setminus \left(\mathcal{D} \cup E^-_{\mathcal{S}}\right)}(S_{i+1}, \tilde{V} \setminus S_{i+1}) &\leq e_{G \setminus \left(\mathcal{D} \cup E^-_{\mathcal{S}}\right)}(S_{i+1}, S_{i} \setminus S_{i+1}) + e_{G \setminus \left(\mathcal{D} \cup E^-_{\mathcal{S}}\right)}(S_{i+1}, \tilde{V} \setminus S_{i}) \\
    &\leq \left( \operatorname{vol}_{G \setminus \left(\mathcal{D} \cup E^-_{\mathcal{S}}\right)}(S_{i}) - \operatorname{vol}_{G \setminus \left(\mathcal{D} \cup E^-_{\mathcal{S}}\right)}(S_{i+1}) \right) + e_{G \setminus \left(\mathcal{D} \cup E^-_{\mathcal{S}}\right)}(S_{i+1}, \tilde{V} \setminus S_{i})
\end{align*}
According to the definition of state \eqref{def:state}, we have that the capacity of the edges $e \in E_{G \setminus \mathcal{D}}(S_{i+1}, \tilde{V} \setminus S_i)$ are saturated and the edges $e \in E_{G \setminus \left(\mathcal{D} \cup E^-_{\mathcal{S}}\right)}(\tilde{V} \setminus S_i, S_{i+1})$ do not carry any flow. Hence the flow on the edges of $E_{G \setminus \left(\mathcal{D} \cup E^-_{\mathcal{S}}\right)}(S_{i+1}, \tilde{V} \setminus S_i)$ must either be sourced at $S_{i+1}$ or enter through the edges $E_{G \setminus \left(\mathcal{D} \cup E^-_{\mathcal{S}}\right)}(S_i \setminus S_{i+1}, S_{i+1})$. Hence, we can bound 
\begin{align*}
    \cc_{G \setminus \left(\mathcal{D} \cup E^-_{\mathcal{S}}\right)}(S_{i+1}, \tilde{V} \setminus S_{i}) &\leq \Delta(S_{i+1}) + \cc_{G \setminus \left(\mathcal{D} \cup E^-_{\mathcal{S}}\right)}(S_{i} \setminus S_{i+1}, S_{i+1}) \\
    &\leq \Delta(S_{i+1}) + \left(\cc_{G \setminus \left(\mathcal{D} \cup E^-_{\mathcal{S}}\right)}(S_{i}) - \cc_{G \setminus \left(\mathcal{D} \cup E^-_{\mathcal{S}}\right)}(S_{i+1})\right).
\end{align*}
Using that the capacity of each edge is $\frac{18}{\phi \psi^2}$ and the second inequality of \eqref{def:Witness}, we obtain the bound
\begin{align*}
    e_{G \setminus \left(\mathcal{D} \cup E^-_{\mathcal{S}}\right)}(S_{i+1}, \tilde{V} \setminus S_{i+1}) &\leq \frac{\phi \cdot \psi^2\cdot \Delta(S_{i+1})}{18} + 2 \cdot \left( \operatorname{vol}_{G \setminus \left(\mathcal{D} \cup E^-_{\mathcal{S}}\right)}(S_{i}) - \operatorname{vol}_{G \setminus \left(\mathcal{D} \cup E^-_{\mathcal{S}}\right)}(S_{i+1}) \right) \\
    &\leq \frac{4 \cdot \phi \cdot \psi \cdot \operatorname{vol}_W(S_{i+1})}{18} + \frac{\phi}{36} \cdot \operatorname{vol}_{G \setminus \left(\mathcal{D} \cup E^-_{\mathcal{S}}\right)}(S_{i+1}) \leq \frac{\phi}{4} \cdot \operatorname{vol}_{G}(S_{i+1}),
\end{align*}
\end{proof}
For item \ref{itm:DEP-lm3}, we remark that item 2 and 4 imply that $\bigcup_{P \in \mathcal{P}} P$ is a sparse cut in $G \setminus \left( \mathcal{D} \cup E_{\mathcal{S}}^-\right)$, i.e. 
\[e_{G \setminus (\mathcal{D} \cup E_{\mathcal{S}}^-)}\left(\bigcup_{P \in \mathcal{P}} P, V \setminus \bigcup_{P \in \mathcal{P}} P\right) \leq |E_r| \leq \frac{\phi}{4} \cdot \sum_{P \in \mathcal{P}} \vol_G(P).\]
Likewise, we find that $\bigcup_{P \in \mathcal{P}} P$ is a sparse cut in $W \setminus \Pi^{-1}\left( \mathcal{D} \cup E_{\mathcal{S}}^-\right)$, i.e. 
\begin{align*}
    e_{W \setminus \Pi^{-1}(\mathcal{D} \cup E_{\mathcal{S}}^-)}\left(\bigcup_{P \in \mathcal{P}} P, V \setminus \bigcup_{P \in \mathcal{P}} P\right) &\leq \frac{\psi}{\phi} \cdot e_{G \setminus (\mathcal{D} \cup E_{\mathcal{S}}^-)}\left(\bigcup_{P \in \mathcal{P}} P, V \setminus \bigcup_{P \in \mathcal{P}} P\right) \\
    &\leq \frac{\psi}{4} \cdot \sum_{P \in \mathcal{P}} \vol_G(P) \leq \frac{\psi}{4} \cdot \sum_{P \in \mathcal{P}} \vol_W(P),
\end{align*}
where we used in the first inequality that the congestion of $\Pi$ is bounded by $\frac{\psi}{\phi}$. Using Lemma \ref{lm:helper}, this implies that $\min\left( \vol_W(\tilde{V}), \sum_{P \in \mathcal{P}} \vol_W(P) \right) \leq \frac{4 \cdot |\Pi^{-1}(\mathcal{D} \cup E_{\mathcal{S}}^-)|}{3 \psi}$. Since for every vertex $v \in \bigcup_{P \in \mathcal{P}} P$ we have for all time $t$ after $t_P$, the time when $P$ is removed, $\Gamma_t(v) \geq \nabla(v)/2$, we can bound 
\[\sum_{P \in \mathcal{P}} \vol_W(P)/2 \leq \|\Delta\|_1 \leq \frac{4}{\psi} \cdot |\Pi^{-1}(\mathcal{D} \cup E_{\mathcal{S}} \cup E_{\mathcal{P}})| < e(W)/2.\]
This implies that $\vol_W(\tilde{V}) \geq e(W)$ and $\sum_{P \in \mathcal{P}} \vol_W(P) \leq \frac{4 \cdot |\Pi^{-1}(\mathcal{D} \cup E_{\mathcal{S}}^-)|}{3 \psi}$. The run-time bound follows immediately from the run-time bound of Lemma \ref{lm:PushPullRelabel} and the fact that $\|\Delta\|_1 = \frac{4}{\psi} \cdot |\Pi^{-1}(\mathcal{D} \cup E_{\mathcal{S}} \cup E_{\mathcal{P}})|$.

\end{proof}

\subsection{Static Expander Decomposition}
\label{subsec:staticExpDecom}
In this section, we discuss how we can use the algorithm $\textsc{BiDirectedExpanderPruning}$ of \Cref{thm:ExpanderPruning} as a subroutine for a static expander decomposition. But before we turn to the outline of the algorithm, we recall the directed version of the cut-matching game. We point out that in the theorem statements of \cite{khandekar2009graph,louis2010cut} there is no mention of \underline{\textbf{fake}} edges as used below but these can be gleaned off the algorithmic description as first observed in \cite{chuzhoy2020deterministic}.

\begin{theorem}[\cite{khandekar2009graph,louis2010cut}]\label{thm:CutMatching}
    Given a directed graph $G=(V, E)$ of $m$ edges and a parameter $\phi$, the cut-matching game takes $O((m \log^3 m) / \phi)$ time and either returns
    \begin{enumerate}
        \item \label{item:CM-thm1} a $O(1/\log^2(m))$-witness $(W, \Pi)$ certifying that $G \cup \mathcal{F}$ is a $\phi$-expander for some set of \underline{\textbf{fake}} edges $\mathcal{F}$ where $|\Pi^{-1}(\mathcal{F})| \leq c \cdot \frac{m}{\log^4(m)}$, where $c > 0$, or
        \item \label{item:CM-thm2} a balanced sparse cut $(A, \bar{A})$ in $G$: $e_G(A, \bar{A}) \leq O\left(\phi \cdot \log^2(m) \cdot \min(\vol_G(A) , \vol_G(\bar{A}))\right)$ such that $\vol_G(A), \vol_G(\bar{A}) = \Omega(c \cdot m/\log^6 m)$.
    \end{enumerate}
\end{theorem}

The statement above slightly deviates from well-known cut-matching game formulations. It is more common that the cut-matching game either certifies that $G$ is a $\phi$-expander or provides a cut that might be unbalanced. But it is straightforward to obtain the formulation in \Cref{thm:CutMatching}. Recall that the cut-matching algorithm attempts to embed a witness using $O(\log^2(m))$ single commodity flows. A cut $(A, \bar{A})$ is provided if the algorithm fails to route one of these network flows. If one uses the push-relabel algorithm for routing these single commodity flows, it is easy to see that one obtains a pre-flow $\ff$ such that all positive excess flow is stuck on the smaller side of the cut and the total amount is at most $\min(\vol_G(A), \vol_G(\bar{A}))$. Thus, one can readily find a source-sink pair matching $F$ of size at most $\min(\vol_G(A), \vol_G(\bar{A}))$ and extend the pre-flow to an actual routing in $G \cup F$. Indeed if the cut $(A, \bar{A})$ is unbalanced, the algorithm picks such a set of fake edges $F$ and routes the remaining excess flow along these edges. It then continues with routing the next single-commodity flow in $G$. If the algorithm eventually manages to embed a witness, the witness will actually be embedded into $G \cup \mathcal{F}$, where $\mathcal{F}$ is the union of all the fake edge sets $F$ added over all rounds.

In \Cref{alg:StaticExpanderDecomposition}, we then use $\textsc{BiDirectedExpanderPruning}$ to remove the set $\mathcal{F}$ of fake edges from $G \cup \mathcal{F}$ making only marginal adjustments to the witness embedding. In the first line of the algorithm, we check if the CutMatching game provides a cut or a witness embedding. If it provides a cut, we recurse on the two sides of the cut. If it provides a witness embedding, we remove the fake edges using the function \textsc{RemoveEdges}. This subroutine initializes an instance of \textsc{BiDirectedExpanderPruning} for the graph and computes a pruning set for the deleted edges $\mathcal{F}$. In the end, the subroutine computes expander decompositions for the pruning sets and combines those into an expander decompositon of the entire graph.

\begin{theorem}\label{thm:StaticExpanderDecomp}
    Given a directed graph $G$, we can compute a $\left(O(\log^{19} m), \lambda, O\left(\frac{1}{\log^4 m}\right)\right)$-expander decomposition $(\mathcal{X}, E_r)$ in run time $O(m \log^{20}(m)/\phi)$ provided $\lambda \leq O(1/\log^{12}(m))$.
\end{theorem}

\begin{proof}
    Consistent with previous sections, we denote the expansion of the witness by $\psi = O(1/\log^2(n))$. We specify $\phi = \frac{32000 \lambda}{\psi^6}$. We will call the $\textsc{CutMatching}$ with $\phi$ and the \textsc{BiDirectedExpanderPruning} with $(\phi, \psi)$. This guarantees that the expander returned by \textsc{BiDirectedExpanderPruning} are $\lambda$-expanders with $(\lambda, \psi^2/2)$ witness. We also note that $\Pi^{-1}(\mathcal{F})$ satisfies the conditions of \Cref{lm:DirectedExpanderDecomp} provided $c > 0$ is small enough.
    \Cref{Def:ED-item1} and \ref{Def:ED-item3} of the Definition of an expander decomposition are immediate from the algorithm. We prove \cref{Def:ED-item2} and the run-time bound by induction on the size of the graph. For the induction step, we consider a call to $\textsc{ExpanderDecomposition}(G)$ where $G$ has $m$ edges. There are two cases to consider, either we find a balanced sparse cut in $G$ and make a recursive call or we don't and remove the fake edges $\mathcal{F}$ at the end.\\
    \\
    In the first case, we can bound (w.l.o.g. we assume that $e_G(A) \geq e_G(\bar{A})$ and thus $e_G(A) \geq m/3$)
    \begin{align*}
        |E^r| &\leq |E^r_1| + |E^r_2| + e_G(A, \bar{A}) \\
        &\leq c_1 \cdot \lambda \cdot \log^{19}(e(A)) \cdot e(A) + c_1 \cdot \lambda \cdot \log^{19}(e(\bar{A})) \cdot e(\bar{A}) + O(\phi \cdot \log^2(m) \cdot m)\\
        &\leq c_1 \cdot \lambda \cdot \log^{18}\left(m\right) \cdot \log\left(\left(1 - \frac{1}{\log^4(m)}\right)m\right) \cdot e(A) + c_1 \cdot \lambda \cdot \log^{19}(m) \cdot e(\bar{A}) + O(\phi \cdot \log^2(m) \cdot m)\\
        &\leq c_1 \cdot \log^{19}(m) \cdot \phi \cdot m - \frac{c_1}{2} \cdot \lambda \cdot \log^{14}\left(m\right) \cdot e(A) + O(\phi \cdot \log^2(m) \cdot m)\\
        &\leq c_1 \cdot \log^{18}(m) \cdot \phi \cdot m
    \end{align*}
    where we used in the second inequality the induction assumption and \cref{item:CM-thm2} of \Cref{thm:CutMatching} and in the last that $c_1 > 0$ is a constant large enough. Similarly, we bound the run-time by the contributions of the call to $\textsc{CutMatching}$ and the two calls to $\textsc{ExpanderDecomposition}$ with $G[A], G[\bar{A}]$. Again by \Cref{thm:CutMatching} and by the induction assumption, we can bound these contributions (using the same argument as above) by 
    \[O(m \cdot \log^{3}(m)/\phi) + c_2 \cdot e(A) \cdot \log^{20}(e(A))/\lambda + c_2 \cdot e(\bar{A}) \cdot \log^{20}(e(\bar{A}))/\lambda \leq c_2 \cdot \log^{20}(m) \cdot m / \lambda,\]
    where $c_2 >0$ ia again a constant large enough. In the second case, we can use \cref{itm:DEP-thm2} and \cref{itm:DEP-thm3} of \Cref{thm:ExpanderPruning} to bound $|E^r| \leq O(\frac{\phi}{\psi} \cdot |\Pi^{-1}(\mathcal{F})|) $. Using \cref{item:CM-thm1} of \Cref{thm:CutMatching}, we find that provided $c$ is small enough $|E^r| \leq \phi \cdot m \leq O(\lambda \cdot \psi^{-6} \cdot m) = O(\lambda \cdot \log(m)^{12} \cdot m)$. While the run-time is determined by the call to $\textsc{CutMatching}$ and the call to $\textsc{RemoveEdges}$. The run-time of the call to $\textsc{CutMatching}$ is again bounded by $m\cdot \log^3(m)/\phi = O(\lambda \cdot \log^{15}(m))$. To bound the run-time of the call to $\textsc{RemoveVertices}$, we point out that the run-time of the update is in $O\left(\frac{h}{\psi^2} \cdot |\Pi^{-1}(\mathcal{F})|\right)$ \footnote{Here we use the fact that at the time we delete $\mathcal{F}$ the pseudo flow $\ff$ is still $\boldsymbol{0}$.} according to \Cref{thm:ExpanderPruning} and that by induction assumption the run-time of the recursive calls in line 15 is bounded by 
    \begin{align*}
        \sum_{P \in \mathcal{P}_1} \frac{|P| \cdot \log^{20}(|P|)}{\lambda} &\leq \frac{\log^{20}(m)}{\lambda} \cdot \sum_{P \in \mathcal{P}_1} |P| \leq \frac{\log^{20}(m)}{\lambda} \cdot \frac{|\Pi^{-1}(|\mathcal{F}|)|}{\psi} \\
        &= O\left(\frac{m \log^{18}(m)}{\lambda}\right),
    \end{align*}
    where we used item three of \Cref{thm:ExpanderPruning} in the last inequality. 
\end{proof}

\begin{algorithm}[H]
    \begin{algorithmic}[1]
    \caption{$\textsc{ExpanderDecomposition}(G)$}\label{alg:StaticExpanderDecomposition}
    \If{$\textsc{CutMatching}(G)$ provides a cut $(A, \bar{A})$}
    \State $(\mathcal{X}_1, E^r_1) \leftarrow \textsc{ExpanderDecomposition}(G[A])$
    \State $(\mathcal{X}_2, E^r_2) \leftarrow \textsc{ExpanderDecomposition}(G[\bar{A}])$
    \State \Return $(\mathcal{X}_1 \cup \mathcal{X}_2, E^r_1 \cup E^r_2 \cup E_G(A,\bar{A}))$
    \Else
    \State $\textsc{CutMatching}(G)$ provides $(W, \Pi)$ embedded into $G \cup \mathcal{F}$
    \State \Return $\textsc{RemoveEdges}(G, W, \Pi, \mathcal{F})$
    \EndIf
    \State 
    \State def $\textsc{RemoveEdges}(G, W, \Pi, \mathcal{D})$
    \Indent
    \State $((\tilde{V}, \tilde{W}, \tilde{\Pi}), \mathcal{P}, E^r_0) \leftarrow \textsc{BiDirectedExpanderPruning}(G, W, \Pi)$
    \State $((\tilde{V}, \tilde{W}, \tilde{\Pi}), \mathcal{P}, E^r_0).\textsc{RemoveEdges}(\mathcal{D})$
    \State $\mathcal{X} \leftarrow \{(\tilde{V}, \tilde{W}, \tilde{\Pi})\}, E^r \leftarrow E^r_0$
    \For{$P \in \mathcal{P}$}
    \State \label{alg:SED-recursion} $(\mathcal{X}_1, E^r_1) \leftarrow \textsc{ExpanderDecomposition}(P)$
    \State $\mathcal{X} \leftarrow \mathcal{X} \cup \mathcal{X}_1$
    \State $E^r \leftarrow E^r \cup E^r_1$
    \EndFor
    \State \Return $(\mathcal{X}, E^r)$
    \EndIndent
    \end{algorithmic}
\end{algorithm}

\subsection{Dynamic Expander Decomposition}
\label{subsec:dynExpanderDecomposition}

In this section, we discuss how we can use the algorithm $\textsc{BiDirectedExpanderPruning}$ and the algorithm $\textsc{ExpanderDecomposition}$ as a subroutines for a dynamic expander decomposition. The algorithm is given in \textsc{DynamicExpanderDecomposition} (see \Cref{alg:DynamicExpanderDecomposition}). We initialize the data structure with an expander decomposition $(\mathcal{X}_0, E^r_0)$. The data structure then initializes in the lines 3-5, for each component of the expander decomposition an instance of \textsc{BiDirectedExpanderPruning}. For any edge deletion $d$, the data structure envokes the function \textsc{RemoveEdges}. This algorithm, first finds the component $X$ such that $d \in E(X)$ and then deletes the edge from that expander component using the functions of \textsc{BiDirectedExpanderPruning}. This might potentially require to prune away additional subgraphs of the expander $X$. We replace $X$ in the expander decomposition $\mathcal{X}$ by the remainder of $X$ and add the cut edges of the pruning cuts to $E^r$. For any of these pruned subgraphs, we run the static expander decomposition to obtain an expander decomposition $(\mathcal{X}_1, E^r_1)$ and add the components of $\mathcal{X}_1$ to $\mathcal{X}$ and the edges of $E^r_1$ to $E^r$. 

\begin{theorem}\label{thm:DED}For every $(\beta, \phi, \psi)$-expander decomposition $(\mathcal{X}, E_r)$ of a directed graph $G$, there is a randomized data structure $\textsc{DynamicExpanderDecomposition}(G)$ (see \Cref{alg:DynamicExpanderDecomposition}). For up to $c \cdot \phi \cdot \psi \cdot e(G)$ calls, where $c$ is a fixed constant, of the form
\begin{itemize}
    \item $\textsc{RemoveEdge}(d)$: where $d \in E(V)$, adds $d$ to $\mathcal{D}$ (initially $\mathcal{D} = \emptyset$)
\end{itemize}
the algorithm explicitly updates the tuple $(\mathcal{X}, E^r)$ after each call, such that thereafter $(\mathcal{X}, E^r)$ is a $\left(4 \cdot \beta, \frac{\phi \psi^6}{32000}, \frac{\psi^4}{800}\right)$-expander-decomposition for $G \setminus \mathcal{D}$ refining the previous. The run-time is bounded by $O\left(\frac{|\mathcal{D}|}{\phi^2} \cdot \log e(G) \cdot \max\left(\log^{19}|\mathcal{D}|, \frac{1}{\psi^2}\right)\right).$
\end{theorem}

\begin{algorithm}[H]
    \begin{algorithmic}[1]
    \caption{$\textsc{DynamicExpanderDecomposition}(\mathcal{X}_0, E^r_0)$}\label{alg:DynamicExpanderDecomposition}
    \State def $\textsc{Init}$
    \Indent
    \State $(\mathcal{X}, E^r) \leftarrow (\mathcal{X}_0, E^r_0)$
    \For{$X \in \mathcal{X}$}
    \State $(V_X, \mathcal{P}_X, E^r_X) \leftarrow \textsc{BiDirectedExpanderPruning}(X, W_X, \Pi_X)$
    \EndFor
    \EndIndent
    \State 
    \State def $\textsc{RemoveEdge}(d)$
    \Indent
    \State Find $(X, W, \Pi) \in \mathcal{X}$ such that $d \in E(X)$.
    \State $((V_X, W_X, \Pi_X), \mathcal{P}_X,E^r_X).\textsc{RemoveEdges}(d)$
    \State Replace $(X, W, \Pi)$ by $(V_X, W_X, \Pi_X)$
    \For{$P$ new in $\mathcal{P}_X$} 
    \State $(\mathcal{X}_1, E^r_1) \leftarrow \textsc{ExpanderDecomposition}(P)$ \label{alg:DynExpDec-line12}
    \State $\mathcal{X} \leftarrow \mathcal{X} \cup \mathcal{X}_1$
    \State $E^r \leftarrow E^r \cup E^r_1$
    \EndFor
    \EndIndent
    \end{algorithmic}
\end{algorithm}

\begin{proof}[Proof sketch]
    Given that we delete at most $c \cdot \phi \cdot \psi \cdot e(G)$, the assumption of \Cref{thm:ExpanderPruning} is satisfied. \Cref{Def:ED-item1} of the definition of the directed expander decomposition is a direct consequence of \cref{itm:DEP-thm1} in \Cref{thm:ExpanderPruning} and the fact that we apply $\textsc{ExpanderDecomposition}$ to all sets that are pruned away. \Cref{Def:ED-item3} is implied by \cref{itm:DEP-thm4} of \Cref{thm:ExpanderPruning} and \Cref{thm:StaticExpanderDecomp}. For \cref{Def:ED-item2} we note that all edges $E^r$ have either been removed during $\textsc{BiDirectedExpanderPruning}$ or $\textsc{ExpanderDecomposition}$. In the first subroutine we removed at most $O\left(|\mathcal{D}|\right)$ according to \cref{itm:DEP-thm2} and \cref{itm:DEP-thm3}, and, according to \Cref{thm:StaticExpanderDecomp}, in the second subroutine at most $\phi \cdot \beta \cdot p,$ where $p$ is the sum of all $|X_1|$ of line 12. \Cref{itm:DEP-thm3} implies that $p$ is bounded by $\frac{4}{\psi} \cdot \left|\Pi^{-1}(\mathcal{D})\right| = 4 \cdot \frac{|\mathcal{D}|}{\phi} = 2 \cdot e(G).$ Similarly the run-time is determined by the edge removals of $\textsc{BiDirectedExpanderPruning}$ and the calls to $\textsc{ExpanderDecomposition}$. Hence, we may bound the run-time by $O\left(\frac{h}{\psi^2 \phi} \cdot |\mathcal{D}|\right) + O\left( q \cdot \log^{20}(q)/\phi\right) = O\left(\frac{|\mathcal{D}|}{\phi^2} \cdot \log e(G) \cdot \max\left(\log^{19}|\mathcal{D}|, \frac{1}{\psi^2}\right)\right).$
\end{proof}

To obtain \Cref{Main-thm}, we combine \Cref{thm:StaticExpanderDecomp} and \Cref{thm:DED}. After $O(\phi \cdot \psi \cdot e(G))$ deletions, we restart the maintenance of the expander decomposition with \Cref{thm:StaticExpanderDecomp}.

%% file: prev_work.tex
\subsection{Previous Work}
\label{subsec:prevWork}

\paragraph{Expander Decompositions for Static Flow Problems.} In static graph settings, expander decompositions have been employed in many recent algorithms for electrical, maximum flow and min-cost flow problems. As mentioned, they were instrumental in the first Laplacian solver \cite{spielman2004nearly} that computes electrical flows, and still are used in recent Laplacian solvers, for example in the recent first almost-linear deterministic Laplacian solver for directed graphs \cite{kyng2022derandomizing}.\\
For the maximum and min-cost flow problems, expander decompositions have been crucial, as seen in \cite{kelner2014almost,van2020bipartite,van2021minimum,bernstein2022deterministic} and the recent development of an almost-linear time algorithm for max flow and min-cost flow in directed graphs \cite{chen2022maximum}. In \cite{chen2022maximum}, the static min-cost flow problem in a directed graph is transformed via advanced convex optimization methods into a dynamic problem in an undirected graph. This dynamic problem is then solved efficiently by a data structure that uses the undirected expander decomposition algorithm from \cite{saranurak2019expander} internally. By simple reductions, the result in \cite{chen2022maximum} also gave the first almost-linear time algorithms for the problems of negative Single Source Shortest Path (SSSP) and bipartite matching, but also to compute expanders in directed graphs. \\
Since the breakthrough result in \cite{chen2022maximum} (and follow-up work \cite{van2023deterministic, kyng2023dynamic, chen2023almost}), a natural new research initiative has emerged: can we solve the min-cost flow problem without relying on advanced convex optimization methods, or put differently, can it be solved with purely \emph{combinatorial} methods? This question aims to further our understanding of the min-cost flow problem by developing a radically different (possibly more accessible) perspective but is also motivated by the quest to find a simpler and more practical algorithm. This initiative has already led to significant achievements, including a near-linear time algorithm for Negative SSSP \cite{bernstein2022negative} and a purely combinatorial approach to bipartite matching \cite{chuzhoy2024faster} that improves current combinatorial approaches in dense graphs, a barrier that stood since the 80s. \\
Directed expander decomposition has emerged as a critical tool in this landscape, exemplified by the work of \cite{bernstein2022negative}, who utilized directed low-diameter decompositions akin to directed expander decompositions, and \cite{chuzhoy2024faster}, who directly applied directed expander decompositions.
The aim of our new accessible directed expander decomposition framework is to further accelerate this essential research initiative, contributing significantly to the field of graph algorithms.

\paragraph{Expander Decompositions for Graph Problems beyond Flows.} In undirected graphs, expander decompositions have also been crucial in all deterministic almost-linear time global min-cut algorithms for undirected graphs \cite{kawarabayashi2018deterministic,saranurak2021simple,li2021deterministic} in computing short-cycle decompositions \cite{chu2020graph,parter2019optimal,liu2019short}, and in finding min-cut preserving vertex sparsifiers \cite{chalermsook2021vertex,liu2020vertex}.\\
It is noteworthy that the above achievements pertain exclusively to undirected graphs. Directed graphs have yet to benefit from the application of expander decompositions.
In the directed setting only considerably less efficient algorithms are known. We hope that our directed expander decomposition framework will facilitate adapting the existing methodologies used in undirected graphs to the directed context, or will inspire novel strategies to address these algorithmic challenges.

\paragraph{Expander Decompositions in Dynamic Graphs.} 
In dynamic graphs, which are characterized by ongoing edge insertions and deletions, expanders have played a significant role in the undirected setting. They have been fundamental in achieving new worst-case update time and derandomization results in dynamic connectivity \cite{wulff2017fully,nanongkai2017dynamic,nanongkai2017dynamicMinimum,chuzhoy2020deterministic}, in single-source shortest paths \cite{chuzhoy2019new,bernstein2020deterministic,chuzhoy2021deterministic,chuzhoy2021decremental,bernstein2022deterministic}, in approximate $(s, t)$-max-flow and $(s,t)$-min-cut algorithms \cite{goranci2021expander}, and in developing sparsifiers resistant to adaptive adversaries \cite{bernstein2020fully, chen2022maximum}. Notably, they were also a key component in the first subpolynomial update time c-edge connectivity algorithm \cite{jin2022fully} and bounded-value min-cut algorithm \cite{jin2024fully}.\\
In the context of directed graphs, there remains a significant gap in our understanding. A notable challenge is the absence of near-linear time solutions for many problems, such as decremental Single-Source Shortest Paths (SSSP). Where solutions do exist, such as for decremental Strongly Connected Components (SCC), they are typically effective only against oblivious adversaries, as highlighted in \cite{bernstein2019decremental}. Furthermore, in scenarios where algorithms are devised to tackle adaptive adversaries, the trade-off is often a drastic reduction in speed, a fact exemplified in \cite{bernstein2020deterministic}. However, the use of directed expander decompositions in these algorithms suggests that enhancing these decompositions could be key to developing faster and more robust algorithms for directed dynamic graphs.